\documentclass[a4paper,11pt]{article}

\usepackage{amssymb}
\usepackage{amsmath}
\usepackage{subfigure}
\usepackage{slashed} 
\usepackage[colorlinks,linkcolor=black,citecolor=blue]{hyperref}

\allowdisplaybreaks

\usepackage{soul}

\usepackage{graphicx}
\usepackage{amsthm}
\usepackage{latexsym}
\usepackage[dvips]{epsfig}
\usepackage{xcolor}

\usepackage{wasysym}
\usepackage{mathrsfs}
\usepackage{bm}
\usepackage{authblk}
\usepackage{slashed}
\usepackage{yhmath} 
\usepackage{stmaryrd}

\theoremstyle{plain}
\newtheorem{proposition}{Proposition}
\newtheorem{lemma}{Lemma}
\newtheorem{theorem}{Theorem}
\newtheorem{assumption}{Assumption}

\newtheorem{definition}{Definition}
\newtheorem{remark}{Remark}

\def\bmg{{\bm g}}

\def\bml{{\bm l}}
\def\bmn{{\bm n}}
\def\bmm{{\bm m}}

\def\bmA{{\bm A}}
\def\bmB{{\bm B}}
\def\bmC{{\bm C}}

\def\bmxi{{\bm \xi}}

\def\bmpartial{{\bm \partial}}

\newcounter{mnotecount}

\newcommand{\mnotex}[1]
{\protect{\stepcounter{mnotecount}}$^{\mbox{\footnotesize $\bullet$\themnotecount}}$ 
\marginpar{
\raggedright\tiny\em
$\!\!\!\!\!\!\,\bullet$\themnotecount: #1} }

\def\a{{\alpha}}
\def\be{{\beta}}
\def\ga{\gamma}

\def\ep{\epsilon}

\def\ka{\kappa}
\def\la{\lambda}

\def\si{\sigma}

\def\ze{\zeta}
\def\ka{\kappa}

\def\al{\alpha}

\def\rh{{\rho}}

\def\io{{\iota}}

\def\MM{{\mathcal{M}}}
\def\NN{{\mathcal{N}}}
\def\mcS{{\mathcal{S}}}
\def\UU{{\mathcal{U}}}
\def\VV{{\mathcal{V}}}

\def\scri{{\mathscr{I}^+}}

\def\edt{{\eth}}

\newcommand{\nr}{\nonumber}
\def\ni{\noindent}

\usepackage[utf8]{inputenc}
\usepackage[english]{babel}
 
\usepackage{multicol}
\usepackage{color}
 
\usepackage{comment}
 
\begin{document}

\title{\textbf{An asymptotic characterisation of the Kerr spacetime}}

\author[1]{R. Sansom\footnote{E-mail address:{\tt
      r.sansom@qmul.ac.uk}}} \author[,1]{J. A. Valiente
  Kroon \footnote{E-mail address:{\tt j.a.valiente-kroon@qmul.ac.uk}}}

\affil[1]{School of Mathematical Sciences,
  Queen Mary, University of London, Mile End Road, London E1 4NS,
  United Kingdom.}

\maketitle

\begin{abstract}
\ni We provide a characterisation of the Kerr spacetime close to future null infinity using the asymptotic characteristic initial value problem in a conformally compactified spacetime. Stewart's gauge is used to set up the past-oriented characteristic initial value problem. By a theorem of M. Mars characterising the Kerr spacetime, we provide conditions for the existence of an asymptotically timelike Killing vector on the development of the initial data by demanding that the spacetime is endowed with a Killing spinor. The conditions on the characteristic initial data ensuring the existence of a Killing spinor are, in turn, analysed. Finally, we write the conditions on the initial data in terms of the free data in the characteristic initial value problem. As a result, we characterise the Kerr spacetime using only a section of future null infinity and its intersection with an outgoing null hypersurface.
\end{abstract}

\tableofcontents

\section{Introduction}

The Kerr family of solutions to the Einstein vacuum equations 
\begin{equation}\label{eq:EVE}
\mathbf{Ric}[\tilde\bmg] =0,
\end{equation}
where $\tilde\bmg$ is a 4-dimensional Lorentzian metric, describes a stationary, axisymmetric rotating black hole. In the standard Boyer-Lindquist coordiantes $(t,r,\theta,\varphi)$ the metric takes the form
\begin{align}
\begin{aligned}\label{eq:KerrBL}
\tilde\bmg_{Kerr} =&  -(1-\frac{2mr}{\rho^2})dt^2+\frac{\rho^2}{\Sigma}dr^2+\rho^2d\theta^2+\left(r^2+a^2+\frac{2m ra^2\sin^2\theta}{\rho^2}\right)\sin^2\theta d\varphi^2\\
&-\frac{4m ra\sin^2\theta}{\rho^2}dtd\varphi
\end{aligned}
\end{align}
where 
\[
\rho^2\equiv r^2+a^2\cos^2\theta, \qquad \Sigma\equiv r^2-2m r+a^2,
\]
$m$ and $a$ are Kerr parameters relating to the mass and angular momentum, respectively. Its stationarity and axisymmetry is captured by the notion of \emph{Killing vectors} which represent the symmetries of the spacetime. By using suitable coordinates, such as the Boyer-Lindquist coordinates, the Killing vectors of the Kerr spacetime can be read from the metric. Inspecting \eqref{eq:KerrBL}, it is seen that the components of the metric are independent of the coordinates $t$ and $\varphi$. The Killing vector field corresponding to stationarity is $\bmpartial_t$ and to axisymmetry is $\bmpartial_\phi$. Moreover, the Kerr spacetime admits so-called ``hidden symmetries'' that cannot be captured by Killing vectors. These ``hidden symmetries'' can be conveniently described using \emph{Killing spinors} which can be thought of as being more fundamental than the Killing vectors. In a series of papers by Bäckdahl and Valiente Kroon \cite{BaeVal10b,BaeVal11b,BaeVal12} the authors showed that a vacuum spactime admitting a Killing spinor, along with suitable conditions on the Weyl curvature and asymptotics, must be isometric to the Kerr spacetime. This characterisation of the Kerr spacetime is based, in turn, on a characterisation given by Mars \cite{Mar99,Mar00} which exploits the structure of the Weyl tensor and its relation to the Killing vectors of Kerr via the \emph{Mars-Simon tensor} ---see \cite{MarPaeSenSim16,Sim84a}. These results have then been generalised to the Kerr-Newman class of spacetimes in e.g. \cite{Won09,ColVal16a,ColVal16b}. The various characterisations of Kerr are used to study open questions about the Kerr spacetime and more general black hole spacetimes. For example, in black hole uniqueness theorems, see \cite{ChrCosHeu12} for a review of this topic, and in the study of the stability of black holes as solutions to the Einstein equations, specifically through the use of symmetries ---see e.g. \cite{AndBaeBlu15}.

\medskip
\ni In this paper, we characterise the Kerr spacetime in the vicinity of null infinity by constructing conditions for the existence of a Killing spinor on characteristic hypersurfaces where one of these hypersurfaces is taken to be future null infinity --- see Figure \ref{Fig:isomorphictokerr}. That is, for a spacetime $(\tilde{\mathcal{M}},\tilde\bmg)$ solving the Einstein vacuum equations \eqref{eq:EVE}, we investigate the conditions to ensure that, close to future null infinity, $(\tilde{\mathcal{M}},\tilde\bmg)$ is isometric to the Kerr spacetime. To do this we assume \emph{asymptotic simplicity} (see Assumption \ref{Assumption:AsymptoticFlatness}) in order to study the geometry of the asmyptotic region. This allows us to perform a conformal compactification of the physical spacetime $(\tilde{\mathcal{M}},\tilde\bmg)$. An \emph{unphysical spacetime} $(\mathcal{M},\bmg)$ is one in which the physical spacetime $(\tilde{\mathcal{M}},\tilde\bmg)$ is conformally embedded through a diffeomorphism $\varphi:\tilde{\mathcal{M}}\rightarrow\mathcal{M}$ such that 
\begin{equation}
\varphi^*\bmg=\Xi^2\tilde\bmg.
\end{equation}
The scalar function $\Xi$ is the \emph{conformal factor} and can be chosen such that the points where $\Xi=0$ represent the points at infinity in $\tilde{\mathcal{M}}$ and $\bmg$ is well-defined. The set of points where $\Xi=0$ in $\mathcal{M}$ is the conformal boundary. In particular, it contains null infinity
\begin{equation}
\mathscr{I}\equiv \big\{p\in\mathcal{M}:\Xi(p)=0, \mathbf{d}\Xi(p)\neq0\big \}.
\end{equation}
 However, in this regime the vacuum Einstein equations \eqref{eq:EVE} are singular at the conformal boundary. One approach to address this difficulty was given by Friedrich \cite{Fri81a} who obtained a regular set of equations for $\bmg$ that are regular even when $\Xi=0$. These equations are the \emph{conformal Einstein field equations}. Crucially, a solution of the conformal Einstein field equations implies a solution to the vacuum Einstien field equations wherever $\Xi\neq0$ ---see \cite{CFEBook} for a comprehensive discussion of the conformal Einstein field equations.

\begin{figure}[h]
\centering
 \def\svgwidth{14pc}
\begingroup%
  \makeatletter%
  \providecommand\color[2][]{%
    \errmessage{(Inkscape) Color is used for the text in Inkscape, but the package 'color.sty' is not loaded}%
    \renewcommand\color[2][]{}%
  }%
  \providecommand\transparent[1]{%
    \errmessage{(Inkscape) Transparency is used (non-zero) for the text in Inkscape, but the package 'transparent.sty' is not loaded}%
    \renewcommand\transparent[1]{}%
  }%
  \providecommand\rotatebox[2]{#2}%
  \newcommand*\fsize{\dimexpr\f@size pt\relax}%
  \newcommand*\lineheight[1]{\fontsize{\fsize}{#1\fsize}\selectfont}%
  \ifx\svgwidth\undefined%
    \setlength{\unitlength}{111.0920036bp}%
    \ifx\svgscale\undefined%
      \relax%
    \else%
      \setlength{\unitlength}{\unitlength * \real{\svgscale}}%
    \fi%
  \else%
    \setlength{\unitlength}{\svgwidth}%
  \fi%
  \global\let\svgwidth\undefined%
  \global\let\svgscale\undefined%
  \makeatother%
  \begin{picture}(1,0.80397513)%
    \lineheight{1}%
    \setlength\tabcolsep{0pt}%
    \put(0,0){\includegraphics[width=\unitlength,page=1]{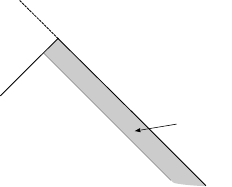}}%
    \put(0.78255923,0.29942052){\color[rgb]{0,0,0}\makebox(0,0)[lt]{\lineheight{1.25}\smash{\begin{tabular}[t]{l}isomorphic to \\Kerr\end{tabular}}}}%
    \put(0.42410273,0.4999247){\color[rgb]{0,0,0}\rotatebox{-0.1340039}{\makebox(0,0)[lt]{\lineheight{1.25}\smash{\begin{tabular}[t]{l}$\mathscr{I}^+$\end{tabular}}}}}%
  \end{picture}%
\endgroup%

\label{Fig:isomorphictokerr}
\caption{The setup of the characteristic initial data with one of the characteristic hypersurfaces being future null infinity $\scri$. The ``nasrrow rectangle'' region shaded in grey is the existence domain of the characteristic initial value problem and the region in which we will characterise the Kerr spacetime.}
\end{figure}

\subsection{The asymptotic characteristic initial value problem for the conformal Einstein field equations}\label{sec:introCIVP}

In order to obtain our characterisation of the Kerr spacetime, we study conditions on initial data for an asymptotic characteristic intial value problem and then employ standard results on the propogation of wave equations to obtain the spacetime result. Therefore, it is crucial to understand the existence and uniqueness domains of such solutions. In this paper, we utilise the analysis of Hilditch, Valiente Kroon and Zhao \cite{HilValZha20b}, who obtained existence of the Conformal Einstein field equations on a ``narrow rectangle'' where one of its long sides is a portion of future null infinity, see Figure \ref{Fig:isomorphictokerr} ---this work is based, in turn, on the analysis of Luk \cite{Luk12}. Following from \cite{HilValZha20b}, we use \emph{Stewart's gauge}, see \cite{SteFri82,Fri91}, which is based on the Newman-Penrose (NP) formalism, see \cite{Ste91,ODo03}. 

In our analysis we will not only need to ensure that there exists a solution to the Conformal Einstein field equations but more general wave equations such as those encoding existence of a Killing spinor in the relevant domain ---see Section \ref{sec:introks} below.  Therefore, we also employ the analysis of \cite{CabChrWaf16} to obtain the existence in the ``narrow rectangle'' of solutions to various wave equations.

\subsection{Killing spinor initial data sets}\label{sec:introks}

As already discussed, the presence of some of the symmetries in the Kerr family of solutions \eqref{eq:KerrBL} is captured by Killing vectors. More generally, imposing that a solution to the vacuum Einstein equations \eqref{eq:EVE} possesses Killing vectors allows symmetry reductions on the equations ---see, for example, \cite{Wei90}. From the viewpoint of the inital value problem, the imposition of symmmetry assumptions should happen only at the level of initial data. This is accomplished by the \emph{Killing initial data equations} \cite{BeiChr97b} which encode the existence of a Killing vector in the developmemnt in the initial data itself. The deeper ``hidden symmetries'' appearing in the Kerr spacetime can be characterised by more fundamental objects known as Killing spinors which are symmetric valence-$2$ spinors $\tilde{\kappa}_{AB}$. The \emph{Killing spinor initial data equations} encode on initial data the existence of a Killing spinor in the initial data's development. For the Einstein vacuum equations \eqref{eq:EVE}, the Killing spinor initial data equations have been derived in \cite{GarVal08a} and for suitable classes of matter in \cite{ColVal16a,ColValRac18}. For the conformal Einstein equations, conditions on spacelike initial data in the unphysical spacetime that give rise to a Killing spinor in the conformally related physical spacetime have been derived in \cite{GasWil22}.

As part of our analysis in this article, we derive the Killing spinor initial data conditions for the conformal Einstein field equations on characteristic initial data in the unphysical spacetime that give rise to a Killing spinor in the conformally related physical spacetime. Moreover, we take one of the null hypersurfaces to be the future null infinity, $\scri$.

The central idea to constructing Killing spinor initial data is to obtain a closed system of homogeneous wave equations that encode the existence of a Killing spinor. For the conformal Einstein field equations this closed homogeneous system was derived in \cite{GasWil22}. Hence, by applying the results outlined in Section \ref{sec:introCIVP}, by ensuring that the equations governing the existence of a Killing spinor hold on the initial data, we can obtain a solution to the Killing spinor initial data equations within some existence domain. In this article, the existence domain will be the ``narrow rectangle'' along $\scri$ ---see Figure \ref{Fig:isomorphictokerr}.

\subsection{The main theorem}

In this section, we present a rough version of the main result of this article, based on the work of Mars \cite{Mar99} and \cite{ColVal16a}.

\begin{theorem}[Main Theorem, Rough version]\label{thm:mainthmrough}
Let $(\tilde{\mathcal{M}},\tilde\bmg)$ be a smooth, asymptotically simple, solution to \eqref{eq:EVE}. In the conformally related spacetime $(\mathcal{M},\bmg)$ satisfying the conformal field equations, let $\mathcal{N}_\star$ denote an outgoing null hypersurface intersecting with $\scri$ at the comapct $2$-dimensional surface $\mathcal{S}_\star$. Then $(\tilde{\mathcal{M}},\tilde\bmg)$ is locally isomorphic to a member of the Kerr family of spacetimes with non-vanishing mass if and only if:
\begin{itemize}
  \item[i.] There exists a Killing spinor $\kappa_{AB}$ in $(\mathcal{M},\bmg)$. That is, $\kappa_{AB}$ satisfies \emph{characteristic conformal Killing spinor initial data} equations on $\NN_\star\cup\scri$.
  \item[ii.] There exists a real positive constant $\mathfrak{l}$ such that  
  \[
  \mathcal{H}^2=\mathfrak{l}\chi^4
  \]
  where $\mathcal{H}$ is a quantity related to the Killing form defined in equation \eqref{eq:defHH2} and $\chi$ is the Ernst potential defined in equation \eqref{eq:defErnstPotential}.
\end{itemize}
 \end{theorem} 
The precise version of Theorem \ref{thm:mainthmrough} is presented in Section \ref{sec:mainthm}. We will rewrite the conditions $\emph{i.}$ and $\emph{ii.}$ in terms of \emph{free data} in the characteristic initial value problem for the conformal Einstein field equations on $\NN\cup\scri$. 

\subsection*{Outline of the paper}

In Section \ref{sec:CIVP}, we briefly introduce the spinorial formalism and conformal Einstein field equations. We introduce the characteristic initial value problem for the conformal Einstein field equations in Stewart's gauge and discuss the various gauge choices and simplifications. Moreover, we prove that a reduced initial data set can be used to generate the full initial data set and discuss asymptotics. 
In Section \ref{Section:KillingSpinorsGeneral}, we introduce the notion of Killing spinors and discuss Mars's and Cole \& Valiente-Kroon's characterisation of the  Kerr spacetime. Additionally, we derive the conformally related characterisation in the unphysical spacetime. That is, we derive the spacetime condtions to characterise the physical Kerr spacetime $(\tilde{\mathcal{M}}_{Kerr},\tilde{\bmg}_{Kerr})$ in the unphysical spacetime $(\mathcal{M},\bmg)$.
In Section \ref{sec:CKSID}, we derive the characteristic Killing spinor initial data in Stewart's gauge. Subsequently, we present some intuitive examples utilising the free data from Section \ref{sec:CIVP}.
Finally, in Section \ref{sec:AsympKerr}, we use the results of Sections \ref{Section:KillingSpinorsGeneral} and \ref{sec:CKSID} to derive asymptotically a characterisation of Kerr near $\scri$. Moreover, we derive conditions on the free data placed on the initial hypersurfaces that will generate the Kerr spacetime in its past development. 



\subsection*{Notation and conventions}
\addcontentsline{toc}{subsection}{Notation and conventions}
In what follows, $(\tilde\MM,\tilde\bmg)$ will denote a physical vacuum spacetime. The metric $\tilde{\bmg}$ is assumed to have signature $(+,-,-,-)$ and admits a spinor structure. Lower case Latin letters ~$a,\,b,\, \ldots$ denote abstract tensorial spacetime indices. Uppercase Latin indices $A,\,B,...$ are used as abstract spinorial indices with indices on complex conjugates of spinors with primes. The conventions for the curvature tensors are fixed by the relation
\begin{align}\label{eq:curvatureconvention}
(\nabla_a \nabla_b -\nabla_b \nabla_a) v^c = R^c{}_{dab} v^d.
\end{align}
We make systematic use of the NP formalism as described, for example,
in~\cite{Ste91,PenRin84}. Definitions and relevant formulas are presented in Appendix \ref{app:NPgauge}. Equality on $\scri$ will be denoted by $\simeq$. For example, the vanishing of the conformal factor $\Xi$ on $\scri$ is given by $\Xi\simeq 0$. Equality on $\NN_\star$ is denoted by $\bumpeq$.\\

Throughout this article, we will project spinors onto the dyad basis $\{o^A,\iota^A\}$. Details are given in Appendix \ref{app:NPgauge}. We will make use of the operators 
\begin{align*}
\eth\eta \equiv \delta\eta + (q\bar\a -p\be)\eta,\\
\bar\eth\eta \equiv \bar\delta \eta - (p\a -q\bar\be)\eta
\end{align*}
for $\eta$ a smooth scalar field of weight $\{p,q\}$. The spin weight of a $\{p,q\}$-scalar is $s=\frac{1}{2}(p-q)$.

\section{The asymptotic characteristic initial value problem for the conformal Einstein field equations}\label{sec:CIVP}

In this section, we provide a brief overview of the spinorial formalism to be used in this article and the characteristic initial value problem in the Newman-Penrose formalism.

\subsection{Spinorial formalism}

Throughout this article, we will make systematic use of spacetime spinors, also called $SL(2,\mathbb{C})$ spinors. In this section, we recapitulate the calculus of spacetime spinors. For more details see \cite{CFEBook, Ste91, PenRin84}. For notational simplicity, we define the following objects in an arbitrary vacuum spacetime $(\mathcal{M},\bmg)$. Analogous to the conventions for tensorial curvature being set through the commutator \eqref{eq:curvatureconvention}, the spinorial curvature conventions are fixed in accordance with \eqref{eq:curvatureconvention} using the spinorial commutator $[\nabla_{AA'},\nabla_{BB'}]$. The spinorial commutator can be expressed in terms of the symmetric operator $\Box_{AB}$ where 
\begin{align}\label{eq:boxdef}
    \square_{AB}\equiv \nabla_{Q'(A}\nabla_{B)}{}^{Q'}
\end{align}
with 
\begin{equation}
    [\nabla_{AA'},\nabla_{BB'}] =\epsilon_{AB}\square_{A'B'}+\epsilon_{A'B'}\square_{AB}.
\end{equation}
where $\epsilon_{AB}$ is the \emph{$\ep$-spinor}, an antisymmetric $2$-spinor satisfying
\[
g_{ABA'B'}=\ep_{AB}\ep_{A'B'},
\]
where $g_{ABA'B'}$ is the spinorial counterpart of the metric $\bmg$. The curvature is then encoded in the action of $\square_{AB}$ on valence-$1$ spinors
\begin{align}
    \begin{aligned}\label{eq:SpinorRicci}
        \square_{AB}\xi_C =& -\Psi_{ABCD}\xi^D+2\Lambda\xi_{(A}\epsilon_{B)C},\\
        \square_{A'B'}=&-\Phi_{CAA'B'}\xi^A.
    \end{aligned}
\end{align}
The formulas \eqref{eq:SpinorRicci} are know as the spinorial Ricci identities where $\Psi_{ABCD}$, $\Phi_{CAA'B'}$ and $\Lambda$ are the Weyl spinor, tracefree Ricci spinor and scalar curvature, respectively.

\medskip
Throughout this article, we will implicitly use the irreducible decomposition of spinors. That is the decomposition of a spinor into its fully symmetric part and sum of products of $\epsilon$ with the trace of the spinor. A particularly useful identity is
\begin{equation}
    \nabla_{Q'A}\nabla_{B}{}^{Q'} = \square_{AB}+\frac{1}{2}\epsilon_{AB}\square
\end{equation}
where $\square\equiv \nabla_{AA'}\nabla^{AA'}$.

\subsection{Conformal geometry and asymptotics}

In this section we summarise the key results in conformal geometry that we will make use of throughout this article. Further details can be found in, for example, Chapters 5 and 10 of \cite{CFEBook}. 

\medskip
The \emph{physical spacetime}, $(\tilde{\mathcal{M}},\tilde{\bmg})$ satisfying the vacuum Einstein equations, \eqref{eq:EVE}, is said to be related to the \emph{unphysical spacetime}, $(\mathcal{M},\bmg)$, through a \emph{conformal transformation}
\begin{equation}
\bmg =\Xi^2\tilde{\bmg},
\label{ConformalRescaling}
\end{equation}
where the scalar $\Xi$ is called the \emph{conformal factor}. The \emph{unphysical spacetime} $(\mathcal{M},\bmg)$ together with the conformal factor $\Xi$ generate a solution to the conformal Einstein field equations as defined in Section \ref{sec:CFE}. Crucially for our analysis, we make the following assumption:

\begin{assumption}
\label{Assumption:AsymptoticFlatness}
    The vacuum spacetime $(\tilde{\mathcal{M}},\tilde{\bmg})$ is \emph{asymptotically simple}. In particular, we assume that the conformal boundary associated to the extension $(\mathcal{ M},\bmg,\Xi)$ admits at least a portion of \emph{future null infinity}, $\mathscr{I}^+$.
\end{assumption}

For a detailed definition and discussion of the notion of asymptotic flatness see, e.g.,  \cite{PenRin86,CFEBook}. The conformal boundary $\mathscr{I}^+$ is an ingoing null hypersurface where
\begin{equation}
\Xi = 0, \qquad \mathbf{d}\Xi\neq0.
\end{equation}
 In the following we ignore any issues regarding the behaviour of the spacetime as one approaches spatial or timelike infinity. 
 
 \medskip
 Finally, since we have assumed that $(\mathcal{\tilde M},\tilde\bmg)$ admits a spinor structure, it follows then that $(\mathcal{M},\bmg)$ inherits a spinorial structure. Following standard notation, in the following let $\Psi_{ABCD}$ and $\Phi_{AA'BB'}$ denote, respectively, the \emph{Weyl spinor} and the \emph{tracefree Ricci spinor} of the \emph{unphysical metric} $\bmg$.  We will also make use of the rescaled Weyl spinor defined by
\begin{equation}
    \phi_{ABCD}\equiv \Xi^{-1}\Psi_{ABCD}.
\end{equation}
As a consequence of the asymptotic flatness of $(\mathcal{\tilde M},\tilde\bmg)$ it follows that the Weyl spinor vanishes at $\mathscr{I}^+$ ---that is,
\[
\Psi_{ABCD} \simeq 0,
\]
where $\simeq$ denotes equality at $\mathscr{I}^+$ ---see \cite{PenRin86,Ste91,CFEBook}, so that
\[
\phi_{ABCD}\simeq O(1). 
\]
We also note the following conformal transformation rules on the $\ep$-spinor:
\begin{align}
\begin{aligned}
\ep_{AB}&=\Xi\tilde{\ep}_{AB}, &\qquad \ep^{AB}&=\Xi^{-1}\tilde{\ep}^{AB}\\
\ep_{A'B'}&=\Xi\tilde{\ep}_{A'B'}, &\qquad \ep^{AB}&=\Xi^{-1}\tilde{\ep}^{A'B'}\\
\end{aligned}\label{eq:epsilonconformaltransformation}
\end{align}

\subsection{The spinorial conformal field equations}\label{sec:CFE}

The conformal Einstein field equations are a reformulation of the Einstein vacuum equations using the conformal decomposition \eqref{ConformalRescaling}. The key property of the conformal Einstein field equatons is that they are regular at null infinity $\scri$ where the conformal factor vanishes, $\Xi=0$. \\

The \emph{standard metric vacuum conformal Einstein field
equations} are given by ---see
\cite{Fri81a,Fri81b,Fri82,Fri83}:
\begin{align*}
& \nabla_{a}\nabla_{b}\Xi +\Xi L_{ab} - s g_{ab}=0 \\
& \nabla_{a}s +L_{ac} \nabla ^{c}\Xi=0, \\
& \nabla_{b}L_{ac}-\nabla_{a}L_{bc} - d_{abcd}\nabla^d{}\Xi =0, \\
 &\nabla_{e}d_{abc}{}^{e}=0 , \\ 
&\lambda - 6 \Xi s + 3 (\nabla_{a}\Xi) \nabla^{a}\Xi=0,
\end{align*}
where $\Xi$ is the conformal factor, $s$ is the \emph{Friedrich scalar} defined by
\begin{equation}\label{eq:sdef}
s\equiv \tfrac{1}{4}\nabla_{a}\nabla^{a}\Xi + \tfrac{1}{24}R\Xi,
\end{equation}
$L_{ab}$ is the \emph{Schouten tensor} defined by
\begin{equation}\label{eq:schouten}
L_{ab}\equiv\tfrac{1}{2}R_{ab}-\tfrac{1}{12}Rg_{ab},
\end{equation}
and $d^{a}{}_{bcd}$ denotes the \emph{rescaled Weyl tensor}, related to the Weyl tensor $C^{a}{}_{bcd}$ through
\[d^{a}{}_{bcd}\equiv \Xi^{-1}C^{a}{}_{bcd}.\]
 The spinorial version of these equations will be useful in what follows. They can be written as
\begin{align}
\begin{aligned}
   & \nabla_{AA'}\nabla_{BB'}\Xi - \Xi \Phi _{ABA'B'} -
  s \epsilon _{AB} \epsilon _{A'B'} + \Xi \Lambda \epsilon _{AB}
  \epsilon _{A'B'}=0 ,
  \\
  & \nabla_{AA'}s + \Lambda \nabla_{AA'}\Xi - \Phi
  _{ABA'B'} \nabla^{BB'}\Xi =0,\\ 
  &\nabla_{A'(A}\Phi _{B)CC'}{}^{A'} - \epsilon _{C(A}
  \nabla_{B)C'}\Lambda + \phi _{ABCD} \nabla^{D}{}_{C'}\Xi=0, \\
  & \Lambda _{C'ABC} =
  \nabla_{DC'}\phi _{ABC}{}^{D}=0, \\ 
  &\lambda -6 \Xi s + 3 (\nabla_{AA'}\Xi)
  \nabla^{AA'}\Xi=0. 
\end{aligned}
\end{align}

\subsection{The past-oriented asymptotic characteristic initial value problem for the conformal Einstein field equations}
\label{Section:StewartGauge}

In this section, we provide a discussion of the past-oriented characteristic initial value problem for the conformal Einstein field equations working in \emph{Stewart's gauge} \cite{SteFri82}. The latter will provide the basic set-up for our analysis.

\subsubsection{Null geometric setup}
In the following let $(\MM,\bmg)$ denote a 4-dimensional spacetime. As we will be working within the setting of a characteristic initial value problem, we further assume the boundary of $\MM$ to have a further component consisting of a transversal outgoing hypersurface $\NN_\star$ intersecting at a 2-dimensional surface $\mcS_\star$ which we assume to be diffeomorphic to $\mathbb{S}^2$. We define a coordinate system $(x^\mathcal{A})$ on $\mcS_\star$ and in a neighbourhood $\UU$ of $\mcS_\star$ in $\MM$ we introduce null coordinates $u$ and $v$ such that
\begin{equation*}
    \scri = \{p\in\UU : v(p)=0\},\qquad \NN_\star =\{p\in\UU : u(p)=0\}.
\end{equation*}
The coordinates $x^{\mathcal{A}}$ are propagated into $\scri$ and $\NN_\star$ by setting them to be constant along the generators of $\scri$ and $\NN_\star$.
We denote the subset of $D^-(\NN_\star\cup\scri)$, the causal past of $\NN_\star\cup\scri$,  where the coordinates $u$ and $v$ are valid by $\VV$. In the following we further assume that the null hypersurfaces of constant $u$ and $v$ intersect at 2-dimensional surfaces, $\mathcal{S}_{u,v}$, each with the topology of $\mathbb{S}^2$. 

\medskip
In the following we will further specialise our gauge choice in $\mathcal{V}$ to a prescription given in \cite{SteFri82} which in the following we refer to as \emph{Stewart's gauge}. This gauge involves, not only, a prescription of coordinates but also of a Newman-Penrose (NP) frame and a specific choice of the scaling in equation \eqref{ConformalRescaling}. 

\medskip
An NP-frame is constructed by choosing vector fields $\bml$ and $\bmn$ to be tangent to the generators of $\scri$ and $\NN_\star$ normalised by $\bmg(\bml,\bmn)=1$. The preservation of this normalisation under Lorentz transformations can be used to set $l_a=\nabla_a u$, where some freedom remains in setting the origin of $u$. We choose vector fields $\bmm$ and $\bar{\bmm}$ tangent to $\mathcal{S}_\star$ and satisfy $\bmg(\bmm,\bar{\bmm})=-1$ and $\bmg(\bmm,\bmm)=0$ where there remains the freedom to perform spin transformations at each point. 
Thus, we consider a NP frame $\{\bml,\, \bmn,\, \bmm,\, \bar{\bmm}\}$ of the form
\begin{subequations}
\begin{align}
& \bml=Q\bmpartial_v , \label{framel}\\
& \bmn=\bmpartial_u+C^{\mathcal{A}}\bmpartial_{\mathcal{A}}, \label{framen}\\ 
& \bmm=P^{\mathcal{A}} \bmpartial_{\mathcal{A}}. \label{framem}
\end{align}
\end{subequations}
where $C^\mathcal{A}=0$ on $\scri$ and $n^a\nabla_a x^A=0$. A schematic depiction of the above geometric set-up is given in Figure \ref{Figure:PastOrientedCIVP}. 

\begin{figure}[t]
\begin{center}
\def\svgwidth{18pc}
\begingroup%
  \makeatletter%
  \providecommand\color[2][]{%
    \errmessage{(Inkscape) Color is used for the text in Inkscape, but the package 'color.sty' is not loaded}%
    \renewcommand\color[2][]{}%
  }%
  \providecommand\transparent[1]{%
    \errmessage{(Inkscape) Transparency is used (non-zero) for the text in Inkscape, but the package 'transparent.sty' is not loaded}%
    \renewcommand\transparent[1]{}%
  }%
  \providecommand\rotatebox[2]{#2}%
  \newcommand*\fsize{\dimexpr\f@size pt\relax}%
  \newcommand*\lineheight[1]{\fontsize{\fsize}{#1\fsize}\selectfont}%
  \ifx\svgwidth\undefined%
    \setlength{\unitlength}{284.92778085bp}%
    \ifx\svgscale\undefined%
      \relax%
    \else%
      \setlength{\unitlength}{\unitlength * \real{\svgscale}}%
    \fi%
  \else%
    \setlength{\unitlength}{\svgwidth}%
  \fi%
  \global\let\svgwidth\undefined%
  \global\let\svgscale\undefined%
  \makeatother%
  \begin{picture}(1,1)%
    \lineheight{1}%
    \setlength\tabcolsep{0pt}%
    \put(0,0){\includegraphics[width=\unitlength,page=1]{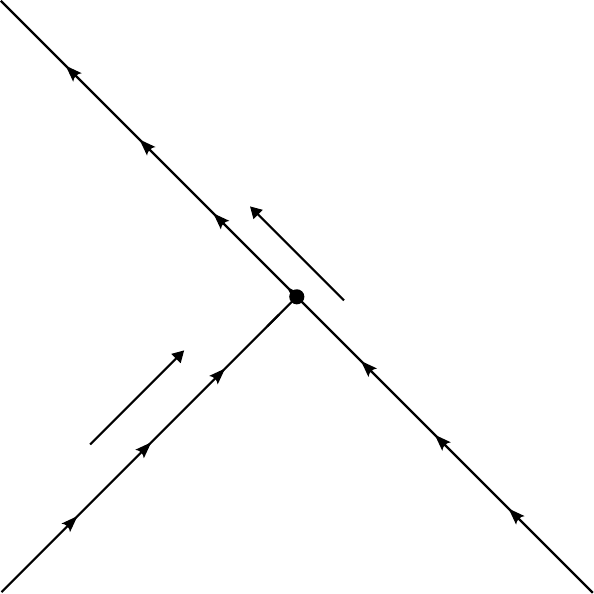}}%
    \put(0.77763108,0.0918816){\color[rgb]{0,0,0}\makebox(0,0)[lt]{\lineheight{1.25}\smash{\begin{tabular}[t]{l}$\scri$\end{tabular}}}}%
    \put(0.14091774,0.0918816){\color[rgb]{0,0,0}\makebox(0,0)[lt]{\lineheight{1.25}\smash{\begin{tabular}[t]{l}$\NN_\star$\end{tabular}}}}%
    \put(0.47418291,0.42756396){\color[rgb]{0,0,0}\makebox(0,0)[lt]{\lineheight{1.25}\smash{\begin{tabular}[t]{l}$\mcS_\star$\end{tabular}}}}%
    \put(0.05824449,0.34404872){\color[rgb]{0,0,0}\makebox(0,0)[lt]{\lineheight{1.25}\smash{\begin{tabular}[t]{l}$v$, $\bml$, $D$\end{tabular}}}}%
    \put(0.50202127,0.5978694){\color[rgb]{0,0,0}\makebox(0,0)[lt]{\lineheight{1.25}\smash{\begin{tabular}[t]{l}$u$, $\bmn$, $\Delta$\end{tabular}}}}%
  \end{picture}%
\endgroup%

\end{center}
\caption{The setup of the coordinates, generators and directional derivatives of the null hypersurfaces $\NN_\star$ and $\scri$ in the past oriented characteristic initial value problem.}
\label{Figure:PastOrientedCIVP}
\end{figure}
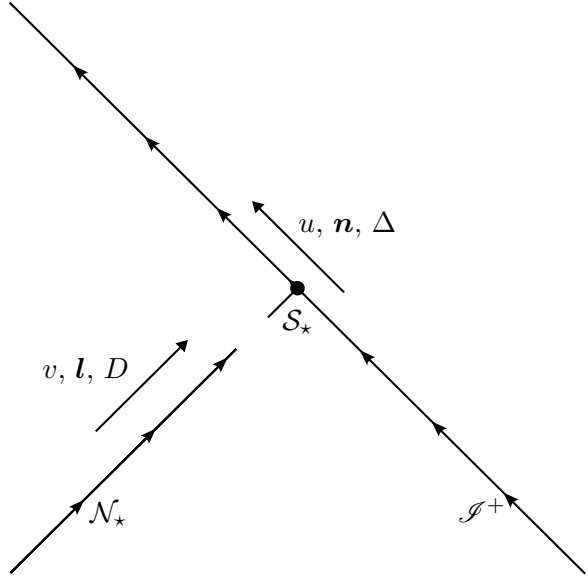

\medskip
In the following subsections we review further aspects of Stewart's gauge.

\subsubsection{Spin connection coefficients and equations for frame coefficients}\label{sec:gaugeconditions}

In this section, we analyse the NP commutators in the NP formalism to arrive at the gauge conditions on the NP spin coefficients implied by Stewart's gauge.

\begin{lemma}[Gauge conditions on the NP spin coefficients]
\label{lem:SpinCoeffCondition}
The NP frame can be chosen such that
\begin{align}
& \rho=\bar{\rho},\ \ \mu=\bar{\mu},\ \ \ep=0,\nonumber \\
& \kappa=\nu=0,\label{spinconnection}\\
& \tau=\bar{\alpha}+\beta, \nonumber
\end{align}
on $D^-(\mathcal{N}\cup \mathscr{I}^+)$.
\end{lemma}

\begin{proof}
Acting on a scalar function $\psi$ the commutators in the NP frame are given by 
\begin{align*}
\Delta D \psi ={}&(\gamma + \bar{\gamma}) D \psi
 + D \Delta \psi
 + (\epsilon + \bar{\epsilon}) \Delta \psi
 -  (\pi + \bar{\tau}) \delta \psi
 -  (\bar{\pi} + \tau) \bar\delta \psi ,\\
\delta D \psi ={}&(\bar{\alpha} + \beta -  \bar{\pi}) D \psi
 + D \delta \psi
 + \kappa \Delta \psi
 -  (\epsilon -  \bar{\epsilon} + \bar{\rho}) \delta \psi
 -  \sigma \bar\delta \psi ,\\
\delta \Delta \psi ={}&- \bar{\nu} D \psi
 -  (\bar{\alpha} + \beta -  \tau) \Delta \psi
 + \Delta \delta \psi
 -  \gamma \delta \psi
 + \bar{\gamma} \delta \psi
 + \mu \delta \psi
 + \bar{\lambda} \bar\delta \psi ,\\
\bar\delta \delta \psi ={}&(- \mu + \bar{\mu}) D \psi
 + (- \rho + \bar{\rho}) \Delta \psi
 + \alpha \delta \psi
 -  \bar{\beta} \delta \psi
 + \delta \bar\delta \psi
 -  \bar{\alpha} \bar\delta \psi
 + \beta \bar\delta \psi .
\end{align*}
The NP commutators applied to the double null coordinates then lead to the following conditions on the NP spin coefficients
\begin{align*}
& \rho=\bar{\rho},\ \ \mu=\bar{\mu},\ \ \ep=-\bar\ep,\nonumber \\
& \kappa=\nu=0,\\
& \tau=\bar{\alpha}+\beta. \nonumber
\end{align*}
Consider now the definition of the combination $\ep-\bar\ep$, namely 
\begin{equation*}
    \bmg(\bmm,\nabla_\bml \bmm)=\ep-\bar\ep.
\end{equation*}
Performing a spin rotation $\bmm\rightarrow\mathrm{e}^{\mathrm{i}\vartheta}\bmm$ on $\bmm$ yields
\begin{equation*}
\bmg(\mathrm{e}^{\mathrm{i}\vartheta}\bmm,\nabla_\bml(\mathrm{e}^{\mathrm{i}\vartheta}\bmm))=\ep-\bar\ep+\mathrm{i}\nabla_\bml \vartheta.
\end{equation*}
Hence, we can use the phase $\vartheta$ to set $\ep-\bar\ep=0$. Together with the above condition, $\ep+\bar\ep=0$, it follows then that $\ep=0$ on the whole existence region.
\end{proof}

\begin{remark} {\em We make the following remarks on Lemma \ref{lem:SpinCoeffCondition}:
    \begin{itemize}
        \item[i.] The conditions given in the lemma differ to those found in \cite{HilValZha20b} because we have `swapped' the roles of $\bml$ and $\bmn$.

        \item[ii.] In this article, we make use of both the NP formalism to single out two null directions and the GHP formalism on the 2-dimensional surfaces. Thus, the final gauge condition in Lemma \ref{lem:SpinCoeffCondition} is enforced to rewrite $\a$ and $\be$ in terms of $\tau$.
    \end{itemize}}
\end{remark}

With the information provided by Lemma \ref{lem:SpinCoeffCondition}, the remaining commutators yield conditions on the coefficients $(Q,C^\mathcal{A},P^\mathcal{A})$ of the frame. More precisely, we have that:
\begin{subequations}
\begin{align}
  &D C^{\mathcal{A}}=(\pi+\bar\tau)P^{\mathcal{A}}
  +(\bar{\pi}+\tau)\bar{P}^{\mathcal{A}} , \label{framecoefficient1} \\
  &D P^{\mathcal{A}}=\rh P^{\mathcal{A}}+\si\bar{P}^{\mathcal{A}}, \label{framecoefficient2} \\
  &\Delta P^{\mathcal{A}}-\delta C^{\mathcal{A}}=(\ga-\bar{\ga}
  -\mu)P^{\mathcal{A}}-\bar{\lambda}\bar{P}^{\mathcal{A}},
  \label{framecoefficient3} \\
  &\Delta Q=(\ga+\bar{\ga})Q,  \label{framecoefficient4}\\
  &\bar{\delta}P^{\mathcal{A}}-\delta\bar{P}^{\mathcal{A}}
  =(\alpha-\bar{\beta})P^{\mathcal{A}}-(\bar{\alpha}-\beta)
  \bar{P}^{\mathcal{A}}, \label{framecoefficient5} \\
  &\delta Q=(\tau-\bar{\pi})Q. \label{framecoefficient6}
\end{align}
\end{subequations}

The Ricci equations and conformal field equations in this gauge are presented in Appendix \ref{app:NPgauge}.
\subsubsection{Conformal gauge condition}

Due to the conformal freedom built into the conformal field equations, we have the following conformal gauge conditions: 

\begin{lemma}
\label{Lemma:ConformalGauge}
Let~$(\tilde{\mathcal{M}},\tilde{\bmg})$ denote an asymptotically
simple spacetime 
satisfying~$\tilde{R}_{ab}=0$ and let~$(\mathcal{M}, \bmg,
\Xi)$ with $\bmg=\Xi^2\tilde{\bmg}$ be a conformal extension for which
the condition $\Xi=0$ describes futue null
infinity~$\mathscr{I}^+$. Given the NP
frame~\eqref{framel}-\eqref{framem}, the conformal factor $\Xi$ can be
chosen so that given a null
hypersurface~$\NN_\star$
intersecting~$\mathscr{I}^+$
on~$\mathcal{S}_\star\approx\mathbb{S}^2$ one has
\begin{align*}
\Lambda =0, \qquad \mbox{in a neighourhood} \quad
\mathcal{V}\quad  \mbox{of} \quad \mathcal{S}_{\star} \quad \mbox{on}
\quad J^-(\mathcal{S}_{\star})
\end{align*}
Moreover,
one has the additional gauge conditions
\begin{subequations}
\begin{align*}
& \Sigma_1=-1, \ \ \mu=\rho=0 \ \ \ \ on\ \ \mathcal{S}_{\star}, \\
& \Phi_{22'}=0 \ \ \ \ on\ \ \mathscr{I}^+, \\
& \Phi_{00'}=0 \ \ \ \ on\ \ \mathcal{N}_\star.
\end{align*}
\end{subequations}
Additionally, the components of the derivative of $\Xi$ satisfy
\[
\Sigma_2=\Sigma_3=\Sigma_4=0 \quad on\ \ \mathscr{I}^+.
\]
\end{lemma}

\begin{proof}
    The proof is analogous to the discussion preceding Lemma 18.2 in \cite{CFEBook}. In the above, $\Sigma_0\equiv D\Xi$ is the derivative in the outgoing direction to $\mathscr{I}^+$ and is chosen to equal $-1$ in this case.
\end{proof}

\subsubsection{Free data in the characteristic problem}

In the present setting, a  solution to the conformal Einstein field equations arises from a prescription of initial data on $\scri\cup\NN_\star$. This initial data consists of the conformal factor, derivatives of the conformal factor, spin coefficients, components of the rescaled Weyl spinor and components of the Ricci spinor. Due to the hierarchical nature of these equations, there exists sets of freely prescribable data that is completely unconstrained and from which all other parts of the data can be calculated. A particular choice of these sets is described by the following Proposition.

\begin{proposition}[Freely specifiable data for the characteristic problem]\label{lem:freedata}
In the gauge given by Lemmas ~\ref{lem:SpinCoeffCondition}
 and ~\ref{Lemma:ConformalGauge} in a neighbourhood of $\mcS_\star$, initial data for the conformal Einstein field equations on $\NN_\star\cup\scri$ can be computed from the reduced initial data set consisting of
\begin{align*}
&P^{A}, \,  \si, \ \Phi_{02},\ \phi_{1},\ \phi_{2}+\bar\phi_{2} \qquad \text{on} \ \mcS_\star \\
&\ga,  \ \phi_{4} \qquad \text{on} \ \scri \\
&  \phi_{0} \qquad \text{on} \ \NN_\star.
\end{align*}
\end{proposition}

\begin{proof}






The strategy of the proof is to analyse the intrinsic part of the conformal filed equations associated to the various parts of the initial hypersurface $\mathscr{I}^+\cup\NN_\star$, starting by the intersection $\mcS_\star\equiv\mathscr{I}^+\cap\NN_\star$. 

\medskip
\noindent
 \textbf{Analysis on $\mcS_\star$.} Given $P^\mathcal{A}$, the intrinsic metric and the intrinsic differential operators $\delta$ and $\bar\delta$ are defined and thus we have the induced connection $\a-\bar\be$. Then $C^\mathcal{A}$ and $Q$ are given on $\mcS_\star$. Moreover,  from the intrinsic equation for $Q$, it follows that  $\tau=\bar\pi$.

\smallskip
Now, the intrinsic components of the first conformal field equation are given by 
\begin{align*}
- \Sigma_1 (\bar{\alpha} +  \beta )+ \Sigma_3 \rho + \Sigma_4 \sigma + \delta \Sigma_1={}&\Xi \Phi_{01},\\
 \Sigma_2( \bar{\alpha} + \beta )-  \Sigma_4 \bar{\lambda} -  \Sigma_3 \mu + \delta \Sigma_2={}&\Xi \Phi_{12},\\
\Sigma_3 (\bar{\alpha} -  \beta) -  \Sigma_1 \bar{\lambda} + \Sigma_2 \sigma + \delta \Sigma_3={}&\Xi \Phi_{02},\\
- \Sigma_4 (\bar{\alpha} - \beta) -  \Sigma_1 \mu + \Sigma_2 \rho + \delta \Sigma_4={}&- s
 + \Xi \Lambda
 + \Xi \Phi_{11}.
\end{align*}
Using the conformal gauge condition, these equations reduce to 
\begin{equation*}
\tau=\la=s=0 \qquad \text{on}\ \mcS_\star.
\end{equation*}
Since $\tau=\bar\al+\be$, and using the above, we can find $\al,\be$ and $\pi$. The intrinsic components of the structure equations are
\begin{align*}
\bar\delta \beta - \delta \alpha= &
 + \Phi_{11}
 -  \alpha \bar{\alpha}
 + 2 \alpha \beta
 -  \beta \bar{\beta},\\
0 =& \Phi_{21},\\
\bar\delta \sigma =& \Phi_{01}
 + 3 \alpha \sigma
 -  \bar{\beta} \sigma ,\\
\end{align*}
From the first and second of these equations, we can obtain $\Phi_{11}$ and $\Phi_{21}$ respectively. The third equation is an equation for  $\Phi_{01}$ if we prescribe $\si$.

\smallskip
The intrinsic components of the third conformal field equation are
\begin{align*}
\delta \Phi_{10}- \bar\delta \Phi_{01}={}& \phi_{2}
 - \bar{\phi}_{2}- 2 \Phi_{01} \alpha + 2 \Phi_{10} \bar{\alpha}-  \Phi_{20} \sigma
+ \Phi_{02} \bar{\sigma},\\
\delta \Phi_{11}
 - \bar\delta \Phi_{02}={}&\bar{\phi}_{3} - 2 \Phi_{02} \alpha
 + 2 \Phi_{02} \bar{\beta},\\
\delta \Phi_{21}- \bar\delta \Phi_{12}={}&- 2 \Phi_{21} \beta+ 2 \Phi_{12} \bar{\beta}.
\end{align*}
Thus, if we prescribe $\si$ and $\Phi_{02}$ then the first of these equations is an equation for $\phi_{2}-\bar\phi_{2}$. The second equation is an equation for $\phi_{3}$.

\smallskip
The remaining quantities not covered by the above considerations are part of the reduced data set in the statement of the theorem.




\medskip
\noindent
\textbf{Analysis on $\scri$.} The intrinsic equations for $P^\mathcal{A}$ and $Q$ are
\begin{align*}
&\Delta P^{\mathcal{A}}-\delta C^{\mathcal{A}}=(\ga-\bar{\ga}
  -\mu)P^{\mathcal{A}}-\bar{\lambda}\bar{P}^{\mathcal{A}},
   \\
  &\Delta Q=(\ga+\bar{\ga})Q.
\end{align*}
Thus, $P^\mathcal{A}$ and $Q$ on $\mathscr{I}^+$ can be obtained by solving these equation with data on $\mathcal{S}_\star$. From the structure equations, we have the following coupled system of intrinsic equations for $\mu$ and $\la$:
\begin{align*}
\Delta \mu=&
 - \lambda \bar{\lambda}
 - \gamma \mu
 - \bar{\gamma} \mu
 - \mu^2 ,\\
 \Delta \lambda  =&
 - 3 \gamma \lambda
 + \bar{\gamma} \lambda
 - 2 \lambda \mu.
 \end{align*}
 Given $\ga$ on $\scri$ we can integrate this system to obtain $\mu$ and $\la$ on $\scri$. Subsequently, the above equations for $P^{A}$ and $Q$ can be integrated on $\scri$. Then the structure equations 
 \begin{align*}
 \delta \gamma -\Delta \beta=&
 - \Phi_{12}
 -  \bar{\alpha} \gamma
 - 2 \beta \gamma
 + \beta \bar{\gamma}
 + \alpha \bar{\lambda}
 + \beta \mu
 + \gamma \tau
 + \mu \tau ,\\
 \bar\delta \gamma- \Delta \alpha  =&
 -  \bar{\beta} \gamma
 -  \alpha \bar{\gamma}
 + \beta \lambda
 + \alpha \mu
 + \lambda \tau
 + \gamma \bar{\tau},
 \end{align*}
 along with the component of the third CFE
 \begin{align*}
  \Delta \Phi_{12}={}& - 2 \Phi_{12} \bar{\gamma}
 - 2 \Phi_{21} \bar{\lambda}
 - 2 \Phi_{12} \mu,\\
 \end{align*}
 give $\a$, $\be$ (and thus $\tau$) and $\Phi_{12}$ on $\scri$. The structure equation
 \begin{align*}
 \Delta \pi =& \Phi_{21}
 -  \gamma \pi
 + \bar{\gamma} \pi
 -  \mu \pi
 -  \lambda \bar{\pi}
 -  \lambda \tau
 -  \mu \bar{\tau}
 \end{align*}
 then yields $\pi$. Accordingly, we see that we will have to specify $\ga$ on $\mathscr{I}^+$. Now, the structure equations
 \begin{align*}
 \delta \tau-\Delta \sigma =&
 - \Phi_{02}
 + \bar{\lambda} \rho
 - 3 \gamma \sigma
 + \bar{\gamma} \sigma
 + \mu \sigma
 -  \bar{\alpha} \tau
 + \beta \tau
 + \tau^2 ,\\
\bar\delta \tau- \Delta \rho =& -2 \Lambda
 -  \gamma \rho
 -  \bar{\gamma} \rho
 + \mu \rho
 + \lambda \sigma
 + \alpha \tau
 -  \bar{\beta} \tau
 + \tau \bar{\tau}
 \end{align*}
 and the component of the third conformal field equation
 \begin{align*}
 \Delta \Phi_{02}- \delta \Phi_{12}={}&- \bar{\phi}_{4'} \Sigma_1
  - 2 \Phi_{12} \beta
 + 2 \Phi_{02} \gamma
 - 2 \Phi_{02} \bar{\gamma}
 - 2 \Phi_{11} \bar{\lambda}
 -  \Phi_{02} \mu\nonumber
 \end{align*}
 can be integrated given $\phi_{4}$, $\Phi_{11}$ and $\Sigma_{1}$ on $\scri$. For $\Sigma_{1}$, we have the component of the first conformal field equation
 \begin{align*}
 - \Sigma_1( \gamma +\bar{\gamma} ) + \Delta \Sigma_1={}&s
 \end{align*}
 and the component of the second conformal field equation that gives $s$ on $\scri$, namely
 \begin{align*}
 \Delta s={}&0.
 \end{align*}
 For $\Phi_{11}$, we have the component of the third conformal field equation 
 \begin{align*}
 \Delta \Phi_{11}+ \Delta \Lambda
 - \delta \Phi_{21}={}&
 -  \Phi_{12} \alpha
 -  \Phi_{21} \bar{\alpha}
 + \Phi_{21} \beta
 -  \Phi_{12} \bar{\beta}
 -  \Phi_{20} \bar{\lambda} - 2 \Phi_{11} \mu.
 \end{align*}
 Thus, given $\phi_{4}$ on $\scri$, we have $\rho$ and $\si$. The final two Ricci components have intrinsic equations derived from the third conformal field equation
 \begin{align*}
 \Delta \Phi_{00} + 2 D \Lambda
 - \delta \Phi_{10}={}&- \bar{\phi}_{2} \Sigma_1
  - 2 \Phi_{01} \alpha
 - 4 \Phi_{10} \bar{\alpha}
 - 2 \Phi_{10} \beta
 - 2 \Phi_{01} \bar{\beta}
 + 2 \Phi_{00} \gamma\nonumber\\
& + 2 \Phi_{00} \bar{\gamma}
 -  \Phi_{00} \bar{\mu}
 + 2 \Phi_{11} \bar{\rho}
 + \Phi_{20} \sigma,\\
 \Delta \Phi_{01}+ \delta \Lambda
 - \delta \Phi_{11}={}&- \bar{\phi}_{3} \Sigma_1
  -  \Phi_{02} \alpha
 - 2 \Phi_{11} \bar{\alpha}
 - 2 \Phi_{11} \beta
 -  \Phi_{02} \bar{\beta}
 + 2 \Phi_{01} \gamma\nonumber\\
& -  \Phi_{10} \bar{\lambda}
 -  \Phi_{01} \mu
 + \Phi_{12} \rho
 + \Phi_{21} \sigma.
 \end{align*}
 The only unknowns in the above expressions are the components of rescaled Weyl for which we have the components of the fourth conformal field equations
 \begin{align}
 \begin{aligned}\label{eq:CFE4onScri}
 \Delta \phi_{0} - \delta \phi_{1}={}&-2 \phi_{1} \beta
 + 4 \phi_{0} \gamma
 -  \phi_{0} \mu
 + 3 \phi_{2} \sigma
 - 4 \phi_{1} \tau,\\
\Delta \phi_{1}- \delta \phi_{2}={}&2 \phi_{1} \gamma
 - 2 \phi_{1} \mu
 + 2 \phi_{3} \sigma
 - 3 \phi_{2} \tau,\\
\Delta \phi_{2}-\delta \phi_{3}={}&2 \phi_{3} \beta
 - 3 \phi_{2} \mu
 + \phi_{4} \sigma
 - 2 \phi_{3} \tau,\\
\Delta \phi_{3}- \delta \phi_{4}={}&4 \phi_{4} \beta
 - 2 \phi_{3} \gamma
 - 4 \phi_{3} \mu
 -  \phi_{4} \tau.
 \end{aligned}
 \end{align}
 Given $\phi_{4}$ on $\scri$, we can integrate this system for $\phi_{0}$,...,$\phi_{3}$ and thus obtain the remaining Ricci components by the above.
 



\medskip
\noindent
\textbf{Analysis on $\NN_\star$.} We have the following intrinsic equations for the components of the frame on $\NN_\star$:
\begin{align*}
&D C^{\mathcal{A}}=(\pi+\bar\tau)P^{\mathcal{A}}
  +(\bar{\pi}+\tau)\bar{P}^{\mathcal{A}} ,  \\
  &D P^{\mathcal{A}}=\rh P^{\mathcal{A}}+\si\bar{P}^{\mathcal{A}}.
\end{align*}
Since $Q=1$ on $\NN_\star$, the equation
\begin{align*}
\delta Q =(\tau-\bar\pi)Q
\end{align*}
means that $\tau=\bar\pi$ on $\NN_\star$.

\smallskip
Next, we have the following structure equations for $\rh$ and $\si$:
\begin{align*}
D \sigma  =& -\Xi\phi_{0}
 +2 \rho \sigma,\\
 D \rho  =& \rho^2
 + \sigma \bar{\sigma}.
\end{align*}
Given $\phi_{0}$, this constitutes a system that can be integrated to obtain $\rho$ and $\si$ which, in turn, which can be used to integrate the equation for $P^\mathcal{A}$.
Then we have the intrinsic equations 
\begin{align*}
 -  D \beta=& \Xi\phi_{1}
 -  \beta \rho
 -  \alpha \sigma
 -  \pi \sigma ,\\
- D \alpha =& \Phi_{10}
 -  \alpha \rho
 -  \pi \rho
 -  \beta \bar{\sigma}
\end{align*}
Along  with the component of the third conformal field equation 
\begin{align*}
D \Phi_{10} ={}&\bar{\phi}_{1} \Sigma_1
 -  \bar{\phi}_{0} \Sigma_3
 + 2 \Phi_{10} \rho+ 2 \Phi_{01} \bar{\sigma}.
\end{align*}
For $\Sigma_{1}$ and $\Sigma_{3}$, we have the intrinsic equations from the first conformal field equations 
\begin{align*}
D \Sigma_1={}&0,\\
 D \Sigma_3={}&\Xi \Phi_{01}+ \Sigma_1 \bar{\pi}.\\
\end{align*}
Thus, integrating the system
\begin{align*}
 -  D \beta=& \Xi\phi_{1}
 -  \beta \rho
 -  \alpha \sigma
 -  \pi \sigma ,\\
 - D \alpha =& \Phi_{10}
 -  \alpha \rho
 -  \pi \rho
 -  \beta \bar{\sigma},\\
D \Phi_{10} ={}&\bar{\phi}_{1} \Sigma_1
 -  \bar{\phi}_{0} \Sigma_3
 + 2 \Phi_{10} \rho+ 2 \Phi_{01} \bar{\sigma},\\
 D \Sigma_3={}&\Xi \Phi_{01}+ \Sigma_1 \bar{\pi} ,\\
 D \phi_{1}- \bar\delta \phi_{0}={}&-4 \phi_{0} \alpha
 + \phi_{0} \pi
 + 4 \phi_{1} \rho,\\
\end{align*}
one obtains $\a$, $\be$ (thus $\bar\pi$ and $\tau$), $\Sigma_{3}$, $\Phi_{10}$ and $\phi_{1}$ on $\NN_\star$. Using the latter, the equation for $C^\mathcal{A}$ can be integrated along $\NN_\star$. The remaining undetermined spin coefficients are $\mu$ and $\la$. We have the intrinsic structure equations
\begin{align*}
\delta \pi -D \mu= &2 \Lambda
 + \Xi\phi_{2}
 + \bar{\alpha} \pi
 -  \beta \pi
 -  \pi \bar{\pi}
 -  \mu \rho
 -  \lambda \sigma ,\\
\bar\delta \pi-  D \lambda =& \Phi_{20}
 -  \alpha \pi
 + \bar{\beta} \pi
 -  \pi^2
 -  \lambda \rho
 -  \mu \bar{\sigma}.
\end{align*}
We also have the coupled components of the third conformal field equation 
\begin{align*}
D \Phi_{02}-\delta \Phi_{01}={}&\phi_{1} \Sigma_3
 -  \phi_{0} \Sigma_2
 - 2 \Phi_{01} \beta
 + 2 \Phi_{01} \bar{\pi}
 + \Phi_{02} \rho+ 2 \Phi_{11} \sigma,\\
D \Phi_{11}+ D \Lambda
 - \bar\delta \Phi_{01}={}&\bar{\phi}_{2} \Sigma_1
 -  \bar{\phi}_{1} \Sigma_3
 - 2 \Phi_{01} \alpha
 + \Phi_{01} \pi
 + \Phi_{10} \bar{\pi}
 + 2 \Phi_{11} \rho
 + \Phi_{02} \bar{\sigma},
\end{align*}
the equation for $\Sigma_{2}$, 
 \begin{align*}
 D \Sigma_2={}&s
 -  \Xi \Lambda
 + \Xi \Phi_{11} + \Sigma_3 \pi +  \Sigma_4 \bar{\pi}
 \end{align*}
 and $s$
 \begin{align*}
 D s={}&- \Lambda \Sigma_1
 + \Phi_{11} \Sigma_1
 -  \Phi_{10} \Sigma_3
 -  \Phi_{01} \Sigma_4,
 \end{align*}
 where we recall that $\Sigma_{3}=\bar\Sigma_{4}$. Finally, we have the equation for $\phi_{2}$
 \begin{align*}
 D \phi_{2}- \bar\delta \phi_{1}={}&-2 \phi_{1} \alpha
 -  \phi_{0} \lambda
 + 2 \phi_{1} \pi
 + 3 \phi_{2} \rho.
 \end{align*}
 Putting this all together, we have the following coupled system of intrinsic equations on $\NN_\star$
 \begin{align*}
 \delta \pi -D \mu= &2 \Lambda
 + \Xi\phi_{2}
 + \bar{\alpha} \pi
 -  \beta \pi
 -  \pi \bar{\pi}
 -  \mu \rho
 -  \lambda \sigma ,\\
\bar\delta \pi-  D \lambda =& \Phi_{20}
 -  \alpha \pi
 + \bar{\beta} \pi
 -  \pi^2
 -  \lambda \rho
 -  \mu \bar{\sigma},\\
 D \Phi_{02}-\delta \Phi_{01}={}&\phi_{1} \Sigma_3
 -  \phi_{0} \Sigma_2
 - 2 \Phi_{01} \beta
 + 4 \Phi_{02} \epsilon
 + 2 \Phi_{01} \bar{\pi}
 + \Phi_{02} \rho+ 2 \Phi_{11} \sigma,\\
D \Phi_{11}+ D \Lambda
 - \bar\delta \Phi_{01}={}&\bar{\phi}_{2} \Sigma_1
 -  \bar{\phi}_{1} \Sigma_3
 - 2 \Phi_{01} \alpha
 + \Phi_{01} \pi
 + \Phi_{10} \bar{\pi}
 + 2 \Phi_{11} \rho
 + \Phi_{02} \bar{\sigma},\\
 D \Sigma_2={}&s
 -  \Xi \Lambda
 + \Xi \Phi_{11} + \Sigma_3 \pi +  \Sigma_4 \bar{\pi},\\
 D s={}&- \Lambda \Sigma_1
 + \Phi_{11} \Sigma_1
 -  \Phi_{10} \Sigma_3
 -  \Phi_{01} \Sigma_4,\\
D \phi_{2}- \bar\delta \phi_{1}={}&-2 \phi_{1} \alpha
 -  \phi_{0} \lambda
 + 2 \phi_{1} \pi
 + 3 \phi_{2} \rho,
 \end{align*}
 which we can integrate to obtain $\mu$, $\la$, $\Sigma_{2}$, $s$, $\Phi_{02}$, $\Phi_{11}$ and $\phi_{2}$ on $\NN_\star$. The remaining equations for the rescaled Weyl components, namely, 
 \begin{align*}
 D \phi_{3}- \bar\delta \phi_{2}={}&
 - 2 \phi_{1} \lambda
 + 3 \phi_{2} \pi
 + 2 \phi_{3} \rho,\\
D \phi_{4}- \bar\delta \phi_{3}={}&2 \phi_{3} \alpha
 - 3 \phi_{2} \lambda
 + 4 \phi_{3} \pi
 + \phi_{4} \rho,
 \end{align*}
 allow to determine $\phi_{3}$ and $\phi_{4}$ on $\mathcal{N}_\star$. The remaining Ricci components are $\Phi_{12}$ and $\Phi_{22}$ for which we have the intrinsic equations 
 \begin{align*}
 D \Phi_{12} +\delta \Lambda
 - \delta \Phi_{11}={}&\phi_{2} \Sigma_3
 -  \phi_{1} \Sigma_2
 -  \Phi_{10} \bar{\lambda}
 -  \Phi_{01} \mu
 + \Phi_{02} \pi
 + 2 \Phi_{11} \bar{\pi}
 + \Phi_{12} \rho\nonumber\\
& + \Phi_{21} \sigma,\\
D \Phi_{22}+ 2 \Delta \Lambda
 - \delta \Phi_{21}={}&\phi_{3} \Sigma_3
 -  \phi_{2} \Sigma_2
 + 2 \Phi_{21} \beta
 -  \Phi_{20} \bar{\lambda}
 - 2 \Phi_{11} \mu
 + 2 \Phi_{12} \pi
 + 2 \Phi_{21} \bar{\pi}\nonumber\\
& + \Phi_{22} \rho.
 \end{align*}
Integrating the first gives $\Phi_{12}$ and the second gives $\Phi_{22}$. Finally, the intrinsic equation for $\ga$ on $\NN_\star$ gives $\ga$ on $\NN$. It then follows from the previous discussion that we must specify $\phi_{0}$ on $\NN_\star$.
\end{proof}

\section{Spacetimes with Killing spinors}
\label{Section:KillingSpinorsGeneral}

In this section we discuss general properties of spacetimes endowed with a Killing spinor from the point of view of the asymptotic characteristic initial value problem.

\subsection{Properties of Killing spinors}

The key tool to be used in this article to characterise the Kerr spacetime is the notion of a Killing spinor. In this section we recall its definition and discuss conformal invariance. 
A \emph{Killing spinor} on a spacetime $(\tilde{\MM},\tilde{\bmg})$ is a symmetric rank 2 spinor
$\tilde{\kappa}_{AB}=\tilde{\kappa}_{(AB)}$ satisfying the \emph{Killing spinor
  equation}
\begin{equation}
\tilde\nabla_{A'(A} \tilde\kappa_{BC)}=0.
\label{eq:PhysicalKillingSpinor}
\end{equation}
Given a Killing spinor $\tilde\kappa_{AB}$, the spinor 
\[
\tilde\xi_{AA'} \equiv
\tilde\nabla^P{}_{A'} \tilde\kappa_{AP}
\]
 is the spinorial counterpart of a Killing vector ---i.e. it satisfies the equation
\[
\tilde\nabla_{AA'} \tilde\xi_{BB'} + \tilde\nabla_{BB'} \tilde\xi_{AA'}=0.
\]
A further consequence of the Killing spinor equation \eqref{eq:PhysicalKillingSpinor} is that $\tilde\kappa_{AB}$ satisfies the wave equation
\[
\tilde\square \tilde\kappa _{BC} -   \Psi_{BCAD}\tilde\kappa ^{AD} =0,
\label{PhysicalKSCandidate}
\]
where $\tilde\square \equiv \tilde\nabla_{PP'}\tilde\nabla^{PP'}$. It is important to observe that while any solution to \eqref{eq:PhysicalKillingSpinor} is a solution to \eqref{PhysicalKSCandidate}, the converse is not true. A solution to equation \eqref{PhysicalKSCandidate} will be said to be a \emph{Killing spinor candidate}. 

\medskip
The notion of Killing spinor is conformally invariant. More precisely, if $\tilde{\kappa}_{AB}$ is a solution to the \emph{physical Killing spinor equation} \eqref{eq:PhysicalKillingSpinor} then 
\[
\kappa_{AB}=\Xi^2 \tilde{\kappa}_{AB}
\]
is a solution to the \emph{unphysical Killing spinor equation}
\begin{equation}
\nabla_{A'(A} \kappa_{BC)}=0.
\label{eq:UnphysicalKillingSpinor}
\end{equation}

\begin{remark}
{\em Observe that the above transformation rules imply that 
\[
\kappa^{AB}=\tilde\kappa^{AB}.
\]}  
\end{remark}

The unphysical Killing spinor $\kappa_{AB}$ satisfies, in turn, the equation
\begin{equation}
\square \kappa _{BC}+4 \Lambda  \kappa _{BC} - \Xi   \phi _{BCAD}\kappa ^{AD} =0.
\label{eq:KSWaveEq}
\end{equation}
Solutions to the wave equation \eqref{eq:KSWaveEq} will also be known as \emph{Killing spinor candidates}. 

\medskip
Since all Killing vectors are also conformal Killing vectors, we have that the Killing vector $\tilde\xi^{AA'}$ on $(\tilde{\mathcal{M}},\tilde\bmg)$ extends, on $(\mathcal{M},\bmg)$ to a conformal Killing vector $\xi^{AA'}$ so that
\[
\tilde{\xi}^{AA'}=\xi^{AA'} \qquad \mbox{on} \qquad \mathcal{M}\setminus\mathscr{I}^+.
\]
Lowering indices in the definition of $\tilde{\xi}^{AA'}$ using \eqref{eq:epsilonconformaltransformation} yields
\begin{align*}
    \begin{aligned}
        \Xi^{-2}\xi_{AA'}=\tilde\xi_{AA'} &= \tilde\nabla_{A'}{}^Q\tilde\ka_{AQ}\\
        &=\tilde\ep^{QP}\tilde\nabla_{A'P}\tilde\ka_{AQ}\\
        &=\Xi\ep^{QP}(\nabla_{A'P}(\Xi^{-2}\ka_{AQ})+\Upsilon_{A'A}\Xi^{-2}\ka_{PQ}+\Upsilon_{A'Q}\Xi^{-2}\ka_{AP})\\
        &=\Xi^{-1}\nabla_{A'}{}^{Q}\ka_{AQ}-3\Xi^{-2}(\nabla_{A}{}^Q\Xi )\ka_{AQ}.
    \end{aligned}
\end{align*}
where we have used that $\tilde\ep^{AB}=\Xi\ep^{AB}$ and have set  $\Upsilon_{AA'}\equiv \Xi^{-1}\nabla_{AA'}\Xi$.   Thus
\begin{equation}
     \xi_{AA'} = \Xi\nabla_{A'}{}^{Q}\ka_{AQ}-3(\nabla_{A'}{}^Q\Xi )\ka_{AQ}
     \label{KSToCKV}
\end{equation}
is a conformal Killing vector. The expression for $\xi_{AA'}$ agrees with the expression for the conformal Killing vector in \cite{GasWil22}. It can be verified by an explicit computation that $\xi_{AA'}$ satisfies the conformal Killing equation
\[
\nabla_{AA'}\xi_{BB'} +\nabla_{BB'}\xi_{AA'} = \frac{1}{2}\nabla_{CC'}\xi^{CC'}\epsilon_{AB}\epsilon_{A'B'}. 
\]
Moreover, one also has that 
\[
\xi^{AA'}\nabla_{AA'}\Xi =\frac{1}{4}\Xi\nabla_{AA'}\xi^{AA'}.
\]

\begin{remark}
{\em As most of our analysis will be concerned with the unphysical spacetime $(\mathcal{M},\bmg)$, in the following, unless otherwise specified, for a Killing spinor candidate we will understand a solution to equation \eqref{eq:KSWaveEq}.}   
\end{remark}

\subsection{Propagation of the Killing spinor equation}\label{sec:KillingSpinorEqPropagation}

Initial data for the Killing spinor candidate, equation \eqref{eq:KSWaveEq} as above,  we want to identify conditions on the data such that if the Killing spinor equation vanishes on the initial hypersurface then the Killing spinor equation also vanishes in some neighbourhood of it. The standard strategy to address this problem is to construct a \emph{propagation system} for the Killing spinor equation ---see e.g. \cite{GarVal08b,BaeVal10a,BaeVal10b,GasWil22}.

\medskip
Following \cite{GasWil22}, the discussion of the propagation of the Killing spinor equation in the \emph{unphysical setting} requires the consideration of a system of wave equations for three spinorial zero-quantities constructed from a Killing spinor candidate $\kappa_{AB}$. Namely, one defines
\[
    H_{A'ABC}\equiv 3\nabla_{A'(A}\ka_{BC)},\qquad B_{ABCD}\equiv \ka_{(A}{}^{Q}\phi_{BCD)Q},\qquad
    F_{A'BCD}\equiv \nabla^Q{}_{A'}B_{QBCD}.
\]
In particular, the condition $H_{A'ABC}=0$ implies that the candidate $\kappa_{AB}$ is, in fact, a Killing spinor in the unphysical spacetime. The condition $B_{ABCD}=0$ is an alignment condition between the Killing spinor and the Weyl tensor ---the so-called \emph{Buchdahl constraint}.

\begin{remark}
{\em The Buchdahl constraint $\ka_{(A}{}^{Q}\phi_{BCD)Q}=0$ implies, in the interior of $(\mathcal{M},\bmg)$, that the spacetime is of Petrov type D or more specialised.}
\end{remark}

A lengthy computation then shows that we have the following system of wave equations for the zero quantities ---see \cite{GasWil22} for details:
\begin{proposition}
\label{Proposition:KillingSpinorPropagationSystem}
Let $\kappa_{AB}$ be a solution to the unphysical Killing spinor candidate equation
\eqref{eq:KSWaveEq}. Then the spinorial zero-quantities fields $H_{A'ABC}$, $B_{ABCD}$
and $F_{A'BCD}$ satisfy the system of wave equations
\begin{subequations}
\begin{align}
    \square H_{B'ABC} =& 6 \Lambda H_{B'ABC} - 12 \Xi F_{B'ABC} -
    12(\nabla^{D}{}_{B'}\Xi) B_{ABCD} \nonumber \\
    &\hspace{4cm}+ 3\nabla_{(A|B'|}Q_{BC)} - 6\Phi_{(A}{}^{D}{}_{|B'}{}^{A'}H_{A'|BC)D}, \label{Eq:WaveEqForH}\\
     \square B_{ABCD} =& 12\Lambda B_{ABCD} - 6\Xi
     \phi_{(AB}{}^{FG}B_{CD)FG} +
     2\nabla_{(A|A'}F^{A'}{}_{|BCD)}, \label{Eq:WaveEqForB}\\
    \square F_{A'BCD} =&   \frac{2}{3}\nabla^A{}_{A'}\nabla_{\bm\xi}\phi_{ABCD} -6 \Lambda F_{A'BCD} - 6 \Phi_{(B}{}^{A}{}_{\vert A'}{}^{B'} F_{B'\vert CD)A} - 9 (\nabla^{A}{}_{A'}\Lambda)B_{BCDA} \nonumber \\
    &\hspace{0cm} + 3B_{(BC}{}^{FG} \phi_{D)AFG} \nabla^{A}{}_{A'}\Xi  + 3B_{AF(BC} \nabla^{FB'}\Phi_{D)}{}^{A}{}_{A'B'}  \nonumber \\
    &\hspace{0cm}- 6 \Xi \phi_{(B}{}^{AFG} \nabla_{\vert GA'\vert}B_{CD)AF} \nonumber\\
    & + \nabla^A{}_{A'}\left(8\Lambda B_{ABCD} - 14\Xi \phi_{(AB}{}^{FG}B_{CD)FG} + \frac{2}{3}(\nabla_{FA'}\phi_{G(ABC})H^{A'}{}_{D)}{}^{FG}\right) \label{Eq:WaveEqForF},
\end{align}
\end{subequations}
where, for convenience, we have set
\[
 Q_{BC}\equiv \frac{1}{2}\nabla^{AA'}H_{AA'BC}.
\]
\end{proposition}

\begin{remark}
    {\em The key structural property of the propagation system \eqref{Eq:WaveEqForH}-\eqref{Eq:WaveEqForF} is that it is \emph{homogeneous} in the sense that 
\[
H_{A'ABC}=0, \qquad B_{ABCD}=0
\qquad F_{A'BCD}=0,
\]    
is a solution to the propagation system. It follows from uniqueness of the solutions to the characteristic initial value problem that this solution is generated by the \emph{trivial data} 
\[
H_{A'ABC}=0, \qquad B_{ABCD}=0
\qquad F_{A'BCD}=0, \qquad \mbox{on} \qquad \mathcal{N}_\star \cup \mathscr{I}^+.
\]
As in the case the Killing spinor candidate equation \eqref{eq:KSWaveEq} it is not necessary to prescribe the initial value of the transversal derivatives.
    }
\end{remark}

The following proposition follows from the preceding discussion identifying the conditions on $\mathcal{N}_\star \cup \mathscr{I}^+$ ensuring the existence of a Killing spinor in $D^-(\mathcal{N}_\star\cup\mathscr{I}^+)$. 

\begin{proposition}
\label{Proposition:BasicKillingSpinorData}
    Assume that initial data for the wave equation \eqref{eq:KSWaveEq} can be found so that 
\begin{equation}
        H_{AA'BC}=0,\qquad  B_{ABCD}=0, \qquad F_{A'BCD}=0 \quad \mbox{on} \quad \mathcal{N}_\star\cup \mathscr{I}^+.
        \label{KillingSpinorConditions}
\end{equation}
Then the resulting $\ka_{AB}$ solving the wave equation \eqref{eq:KSWaveEq} is a Killing spinor. 
\end{proposition}

In Sections \ref{sec:CKSID} we will analyse the interdependence of the Killing spinor conditions \eqref{KillingSpinorConditions} with the aim of obtaining a reduced set of conditions on the characteristic data ensuring the existence of a Killing spinor.

\subsection{The conformal Killing vector} \label{sec:CKV}
The conformal Killing vector associated to the (unphysical) Killing spinor via equation \eqref{KSToCKV} provides important information about the behaviour of $\kappa_{AB}$ in a neighbourhood of $\scri\cup \NN_\star$. As we are interested in a characterisation of the Kerr spacetime, it will be required that the spinor $\xi^{AA'}$ at $\scri$ is parallel to the generators of the conformal boundary ---consistent with the presence, in the Kerr spacetime, of a Killing vector which is asymptotically a time translation. 

\medskip
Evaluating equation \eqref{KSToCKV} on $\scri$ where, by definition, one has that $\Xi\simeq0$ yields
\[
\xi_{AA'}\simeq - 3(\nabla^Q{}_{A'} \Xi) \kappa_{AQ}. 
\]
From Lemma \ref{Lemma:ConformalGauge} it follows then that $\nabla_{AA'}\Xi \simeq -\iota_A \bar{\iota}_{A'} $, so that, in turn, one has
\begin{equation}
\xi_{AA'}\simeq 3( \kappa_2 o_A\bar{\iota}_{A'} -\kappa_1 \iota_A \bar{\iota}_{A'}).
\label{KVonScri}
\end{equation}

\medskip
In our gauge, the NP vector $\bmn$ is tangent to the generators of $\scri$. It follows then from  the discussion in the first paragraph in this section that
\[
\bmxi \simeq \bmn \simeq \bmpartial_u. 
\]
Let $n^{AA'}$ denote the spinorial counterpart of $\bmn$. As $n^{AA'}=\iota^A\bar{\iota}^{A'}$ one can choose the normalisation of $\bmxi$ so that 
\[
\xi^{AA'}\simeq \iota^A\iota^{A'}.
\]
Comparing the latter with equation \eqref{KVonScri} one concludes that
\begin{equation}
\kappa_1 \simeq -\frac{1}{3}, \qquad \kappa_2\simeq 0.
\label{ComponentsKappaScri}
\end{equation}
The component $\kappa_0$ on $\scri$ remains unconstrained by the above considerations. 

\medskip
Now, let
\[
\xi_{AA'} = \xi_{00'} \iota_A\bar{\iota}_{A'} + \xi_{01'} \iota_A \bar{o}_{A'} + \xi_{10'} o_A \bar{\iota}_{A'} + \xi_{11'}o_{A}\bar{o}_{A'}. 
\]
For future use, we provide the detailed expressions of the components of the spinor $\xi_{AA'}$. One has that
\begin{subequations}
\begin{align}
\xi_{00'}={}&\kappa_{1} (2 \Xi \rho + 3 D \Xi)
 + \kappa_{0} (\Xi \pi - 3 \bar\delta \Xi)
 + \Xi (- D \kappa_{1} + \edt' \kappa_{0}),\\
- \xi_{01'}={}&\Xi \kappa_{2} \sigma
 - 2 \Xi \kappa_{1} \tau
 + \kappa_{0} (2 \Xi \gamma -  \Xi \mu + 3 \Delta \Xi)
 -  \Xi \Delta \kappa_{0}
 - 3 \kappa_{1} \delta \Xi
 + \Xi \edt \kappa_{1},\\
- \xi_{10'}={}&\Xi \kappa_{0} \lambda
 -  \Xi \kappa_{2} \rho
 - 3 \kappa_{2} D \Xi
 + \Xi D \kappa_{2}
 + \kappa_{1} (-2 \Xi \pi + 3 \bar\delta \Xi)
 -  \Xi \edt' \kappa_{1},\\
\xi_{11'}={}&\kappa_{1} (2 \Xi \mu - 3 \Delta \Xi)
 + \kappa_{2} (\Xi \tau + 3 \delta \Xi)
 + \Xi (\Delta \kappa_{1} -  \edt \kappa_{2}).
\end{align}
\end{subequations}
In particular, at $\scri$ one has that 
\begin{align*}
\xi_{00'}\simeq{}&3 \kappa_{1} D \Xi ,\\
\xi_{01'}\simeq{}&0,\\
\xi_{10'}\simeq{}&3 \kappa_{2} D \Xi ,\\
\xi_{11'}\simeq{}&0.
\end{align*}

\medskip
Applying Lemma \ref{Lemma:ConformalGauge} with equation \eqref{ComponentsKappaScri} one has that
\begin{align}
\xi_{00'}\simeq{}&-D\Xi ,\\
\xi_{01'}\simeq{}&0,\\
\xi_{10'}\simeq{}&0 ,\\
\xi_{11'}\simeq{}&0.
\end{align}

\subsection{An algebraic Killing spinor candidate}
\label{Section:AlgebraicKillingSpinorCandidate}
In \cite{BaeVal11b} a formula has been given that allows to compute an \emph{algebraic Killing spinor candidate} in terms of the various components of the Weyl curvature ---see Proposition 3.1 in that reference. This expression is computable in a generic spacetime and has the property of providing an expression for a Killing spinor if the spacetime in of Petrov type D. More precisely, one has the following:

\begin{definition}
\label{Definition:KSC}
Let $\zeta_{AB}$ be a symmetric spinor satisfying
\[
\zeta_{AB}\neq 0, \qquad \psi^{-1}\Psi_{PQRS}\zeta^{PQ}\zeta^{RS} -\frac{1}{6}\zeta_{PQ}\zeta^{PQ}\neq 0,
\]
and $\Xi$ denote the conformal factor connection the unphysical spacetime to a vacuum physical one. In terms of the above, the symmetric spinor given by 
\begin{equation}
\breve{\kappa}_{AB}\equiv \psi^{-1/3}\Theta^{-1/2}\Xi^2\big(\psi^{-1} \Psi_{ABPQ}\zeta^{PQ}-\frac{1}{6}\zeta_{AB} \big),
\label{Formula:KSCandidate}
\end{equation}
with 
\begin{eqnarray*}
&& \Theta \equiv \psi^{-1}\Psi_{PQRS}\zeta^{PQ}\zeta^{RS}-\frac{1}{6}\zeta_{PQ}\zeta^{PQ},\\
&& \psi\equiv  -18\mathcal{J}\mathcal{I}^{-1},
\end{eqnarray*}
where
\[
\mathcal{I}\equiv \frac{1}{2}\Psi_{ABCD}\Psi^{ABCD}, \qquad \mathcal{J}\equiv \frac{1}{6}\Psi_{ABCD}\Psi^{CDEF}\Psi_{EF}{}^{AB},
\]
and $\Xi$ is the conformal factor connecting to the physical spacetime
is called an algebraic Killing spinor candidate.
\end{definition}

Given an algebraic Killing spinor candidate one has the following:

\begin{proposition}
If an unphysical spacetime $(\mathcal{M},\bmg)$ is of  Petrov type D and conformal to a vacuum spacetime and $\zeta_{AB}$ is a symmetric spinor satisfying the properties of Definition \ref{Definition:KSC}, then the spinor $\breve{\kappa}_{AB}$ is a Killing spinor. The resulting Killing spinor is independent of the choice of $\zeta_{AB}$. 
\end{proposition}

\begin{remark}
\label{Remark:Data}
{\em  For a suitable choice of $\zeta_{AB}$, Formula \eqref{Formula:KSCandidate} can be used to compute characteristic initial data for the Killing spinor equation on $\mathcal{N}_\star\cup\scri$. This data is \em{correct} in the sense that it provides the right values for the components of the Killing spinor in the case that the spacetime is endowed with a Killing spinor. }
\end{remark}

\subsection{A Killing spinor characterisation of Kerr}\label{sec:Kerrcharac}
The interest on the notion of Killing spinors in this article stems from the fact that they provide a characterisation of the Kerr family of spacetimes ---see e.g. \cite{ColVal16a}. In this section, we outline how the Kerr family can be characterised by Killing spinors. 

\medskip
 Let $\tilde\xi^a$ be a Killing vector on the spacetime $(\tilde\MM,\tilde\bmg)$ with spinorial counterpart $\tilde\xi^{AA'}$. The \emph{Ernst form} of the Killing vector $\tilde\xi^a$ is the covector  
\begin{equation*}
    \tilde{\chi}_a\equiv 2\tilde{\xi}^b \tilde{\mathcal{H}}_{ba},
\end{equation*}
where $\tilde{\mathcal{H}}_{ab}$ is the \emph{Killing form} of $\tilde\xi^a$ ---that is
\[
\tilde{\mathcal{H}}_{ab} \equiv  \tilde\nabla_{[a}\tilde\xi_{b]} = \tilde\nabla_a \tilde\xi_b.
\]
The Ernst form is closed and, thus, on simply connected domains also exact. The \emph{Ernst potential} $\tilde\chi$  is the scalar field such that
\begin{equation*}
    \tilde\chi_a=\tilde\nabla_a \tilde\chi.
\end{equation*}
Let $\tilde{H}_{AA'BB'}$ denote the spinorial counterpart of the Killing form $\tilde{\mathcal{H}}_{ab}$. As a consequence of the antisymmetry of the latter one has that 
\[
\tilde{H}_{AA'BB} = \tilde\eta_{AB}\epsilon_{A'B'}+ \bar{\tilde\eta}_{A'B'}\epsilon_{AB},
\]
where $\tilde\eta_{AB}=\tilde\eta_{(AB)}$. It can be readily verified that
\begin{equation}
\tilde\eta_{AB}=\frac{1}{2}\tilde\nabla_{AA'}\tilde\xi_B{}^{A'}.
\label{Definition:TildeEta}
\end{equation}
Finally, define 
\[
\tilde{\mathcal{H}}^2\equiv 8\tilde\eta_{AB}\tilde\eta^{AB}.
\]

\medskip
Of particular importance in the characterisation of the Kerr spacetime is an \emph{alignment condition} between the Weyl tensor and the Killing form. More precisely, we say that the Weyl tensor and the Killing spinor are aligned if there exists a scalar function $\tilde{H}$ such that the spinor $\Psi_{ABCD}$ and $\tilde{\eta}_{AB}$ are related via
\begin{equation}
\Psi_{ABCD} = 2\tilde{H} \tilde\eta_{(AB}\tilde\eta_{CD)}.
\label{PhysicalAlignmentCondition}
\end{equation}

\medskip
In \cite{ColVal16a} the following theorem was proven characterising the Kerr spacetime.

\begin{theorem}
\label{Theorem:SpinorialCharacterisationKerr}
Let $(\tilde{\mathcal{M}},\tilde\bmg)$ denote a smooth vacuum spacetime endowed with
a Killing spinor $\tilde\kappa_{AB}$ satisfying $\tilde\kappa_{AB}\tilde\kappa^{AB}\neq 0$,
such that the spinor $\tilde\xi_{AA'}\equiv \tilde\nabla^B{}_{A'}\tilde \kappa_{AB}$ is
Hermitian. Then there exist two complex constants $\mathfrak{l}$ and
$\mathfrak{c}$ such that 
\begin{equation}
\tilde{H} = \frac{6}{\mathfrak{c}-\tilde{\chi}}, \qquad \tilde{\mathcal{H}}^2 = - \mathfrak{l}( \mathfrak{c}-\tilde\chi)^4.
\label{Condition:Mars}
\end{equation}
\end{theorem}

\begin{remark}
{\em The values of the constants $\mathfrak{l}$ and $\mathfrak{c}$ can be used to characterise the Kerr solution. For example, in \cite{ColVal16a} it was shown that if the spacetime admits an asymptotically stationary flat end then a sufficient condition to single out Kerr is that $\mathfrak{c}=1$ and $\mathfrak{l}$ is a positive real number. Depending on the context, other sufficient conditions may be more convenient.}
\end{remark}

\subsection{Computing the Ernst form and potential}
In this section we note some useful formulae which allow to compute the Ernst form and Ernst potential in terms of the Killing spinor and concomitants.

\medskip
We begin by noting the following:

\begin{lemma}
\label{Lemma:Concommitants}
Let $(\tilde{\mathcal{M}},\tilde\bmg)$ be a vacuum spacetime endowed with a Killing spinor $\tilde{\kappa}_{AB}$. Then
\begin{subequations}
\begin{eqnarray}
&& \tilde\eta_{AB}= -\frac{3}{4}\Psi_{ABCD}\tilde{\kappa}^{CD}, \label{EtaConcomitants} \\
&& \tilde{\chi}_{AA'}= 2 \tilde{\xi}_{QA'}\tilde{\eta}_A{}^Q, \label{ErnstFormConcomitants}\\
&& \tilde{\chi}= \frac{9}{2}\tilde\kappa^{AB}\tilde\kappa^{CD}\Psi_{ABCD}, \label{ErnstPotentialConcomitants}
\end{eqnarray}
\end{subequations}
where $\tilde\chi_{AA'}$ is the spinorial counterpart of the Ernst potential $\tilde\chi_a$.
\end{lemma}

\begin{proof}
Equations \eqref{EtaConcomitants} and \eqref{ErnstFormConcomitants}. Equation \eqref{ErnstPotentialConcomitants} is new and thus we provide a brief discussion of its derivation.  

\smallskip
Starting from equation \eqref{ErnstFormConcomitants}, substituting equation \eqref{EtaConcomitants} one obtains an alternative expression for the Ernst potential ---namely
\[
\tilde\chi_{AA'} = 3 \tilde{\kappa}^{CD} \tilde\xi^B{}_{A'}\Psi_{ABCD}.
\]
On the other hand, a direct computation taking into account that $\tilde{\kappa}_{AB}$ is a Killing spinor yields
\begin{equation}
\tilde{\nabla}_{PP'}\Big( \tilde\kappa^{AB}\tilde\kappa^{CD}\Psi_{ABCD} \Big) = -\frac{4}{3}\tilde\kappa^{BC}\tilde\xi^A{}_{P'}\Psi_{PABC} + \tilde\kappa^{AB}\tilde\kappa^{CD}
\tilde{\nabla}_{PP'}\Psi_{ABCD}.
\label{ErnstFormIntermmediate}
\end{equation}
Now, observe that as a consequence of the Bianchi identity one has that
\[
\tilde{\nabla}_{PP'} \Psi_{ABCD} = \tilde\nabla_{P'(P}\Psi_{ABCD)}. 
\]
Furthermore, starting from the Buchdahl constraint
\[
\tilde\kappa_{(A}{}^Q\Psi_{BCD)Q}=0,
\]
contracting with $\tilde\kappa_{Q}{}^P$ and differentiating with respect to $\tilde\nabla_{FF'}$ one arrives to the identity
\[
\tilde{\kappa}^{AF}\tilde{\kappa}^{PC}\tilde\nabla_{BF'}\Psi_{ACFP}= 2 \tilde{\kappa}^{FP}\tilde{\xi}^A{}_{F'}\Psi_{ABFP}.
\]
The formula for $\xi$ can the be verified by substituting the previous expression in \eqref{ErnstFormIntermmediate}.
\end{proof}

\subsection{Conformal transformation formulae}
The expressions given in Lemma \ref{Lemma:Concommitants} are good for suggesting conformal analogues of $\tilde{\eta}_{AB}$, $\tilde{\chi}_{AA'}$ and $\tilde{\chi}$. The resulting \emph{unphysical expressions} can be used, in turn to evaluate the conditions given in Theorem \ref{Theorem:SpinorialCharacterisationKerr} in the unphysical spacetime $(\mathcal{M},\bmg)$.

\medskip
Now, recall that as a consequence of the \emph{Peeling Theorem}, we have that $\Psi_{ABCD}=O(\Xi)$. Moreover, we restrict our attention to Killing spinor such that $\kappa_{AB}=O(1)$ near the conformal boundary. It then follows from Lemma \ref{Lemma:Concommitants} that
\[
\tilde\eta_{AB} = -\frac{3}{4}\Psi_{ABCD}\kappa^{CD}=O(\Xi).
\]
The latter suggest defining an \emph{unphysical field} $\eta_{AB}$ via
\[
\eta_{AB}\equiv \Xi^{-1}\tilde\eta_{AB}
\]
so that $\eta_{AB}=O(1)$. In particular, since $\phi_{ABCD}\equiv \Xi^{-1}\Psi_{ABCD}$ one can write
\[
\eta_{AB}=-\frac{3}{4}\phi_{ABCD}\kappa^{CD}.
\]
Also, observe that the alignment condition \eqref{PhysicalAlignmentCondition} can be rewritten as
\begin{equation}
\phi_{ABCD} = 2 H \eta_{(AB}\eta_{CD)}, \qquad H\equiv \Xi \tilde{H},\label{eq:unphysicalalignmentcondition}
\end{equation}
so that $H=O(1)$ and $\tilde{H}=O(\Xi^{-1})$. Moreover, defining
\begin{equation}\label{eq:defHH2}
\mathcal{H}^2 \equiv 8\eta_{AB}\eta^{AB}
\end{equation}
one readily finds that
\[
\mathcal{H}^2 = \Xi^{-4}\tilde{\mathcal{H}}^2.
\]

\medskip
Similar arguments to the above show that 
\[
\tilde{\chi}= O(\Xi),
\]
so that, consequently, one can define the unphysical field
\begin{equation}\label{eq:defErnstPotential}
\chi\equiv \Xi^{-1}\tilde{\chi}.
\end{equation}

\medskip
From the previous discussion it follows then the following:

\begin{lemma}
\label{Lemma:AsymptoticSimplicity}
Let $(\tilde{\mathcal{M}},\tilde{\bmg})$ satisfy the assumptions of Theorem \ref{Theorem:SpinorialCharacterisationKerr}. If, in addition, one has that Assumption \ref{Assumption:AsymptoticFlatness} holds and $\kappa_{AB}=O(1)$ near $\mathscr{I}^+$, then necessarily
\[
\mathfrak{c}=0.
\]
\end{lemma}

\begin{proof}
    The result follows directly by inspection. 
\end{proof}

From the previous lemma, it follows that under our assumptions the \emph{unphysical version} of the conditions \eqref{Condition:Mars} are 
\[
H= 6/\chi, \qquad \mathcal{H}^2 = -\mathfrak{l} \chi^4.
\]
As it will be seen in a latter section, the above expressions can be conveniently computed in terms of asymptotic data.

\section{Characteristic conformal Killing spinor initial data}\label{sec:CKSID}

In this section we investigate the restrictions on asymptotic characteristic initial data so that the resulting spacetime is endowed with a Killing spinor. The strategy to be pursued is as follows: in first instance we construct a Killing spinor candidate in $D^-(\NN_\star\cup \mathscr{I}^+)$ ---i.e. a solution to equation \eqref{eq:KSWaveEq}. We then proceed to identify the conditions on the characteristic initial data for the Killing spinor ensuring that the latter is, in fact, an actual Killing spinor. Finally, we reformulate, as much as possible, these conditions in terms of conditions on the characteristic initial data for the conformal Einstein field equations. 

\medskip
Consistent with the discussion of subsection \ref{sec:CKV}, we make the following assumptions on the Killing spinor candidate $\kappa_{AB}$:

\begin{assumption}
\label{Assumption:KSonScri}
The components of $\kappa_{AB}$ satisfy 
\[
\kappa_1 \simeq -\frac{1}{3}, \qquad \kappa_2\simeq 0.
\]
\end{assumption}

\ni We will also make the following assumption on $\NN_\star$:

\begin{assumption}\label{assump:BonNstar}
    The Buchdahl constraint is satisfed on $\NN_\star$. That is,
    \[
    B_{ABCD}\bumpeq 0.
    \]
\end{assumption}

\begin{remark}
    {\em A constructive way to verify the Buchdahl constraint on $\mathcal{N}_\star$ is to make use of the algebraic Killing spinor candidate formula, equation \eqref{Formula:KSCandidate}, in Definition \ref{Definition:KSC}.}
\end{remark}

\ni The main result of this section can be summarised as follows:

\begin{proposition}[Characteristic conformal Killing spinor initial data] 
\label{prop:CCKSID}
    Let the conformal gauge condition of Proposition \ref{Lemma:ConformalGauge} hold. Moreover, let the following conditions be satisfied:
     \begin{itemize}
        \item[i. ] On $\mathcal{S}_\star$, one has that $\phi_3=0$ and $\kappa_0$ satisfies the equations
        \begin{align*}
        \begin{aligned}
        \edt \kappa_{0} -\frac{2}{3} \sigma =& 0,\\
        \frac{1}{3}\phi_{0}
        +\kappa_{0} \phi_{1} =& 0, \\
        \frac{2}{3} \phi_{1}
         +3 \kappa_{0} \phi_{2} =& 0,\\
        \end{aligned}
        \end{align*}
        with $\kappa_0$ propagated along $\scri$ according to the transport equation
        \begin{equation}
        \Delta \kappa_{0} - \frac{2}{3} \tau- 2 \kappa_{0} \gamma\simeq 0.
        \end{equation}
        Moreover, the dyad components of $B_{ABCD}$ satisfy 
        \begin{align}
        \begin{aligned}
         & D B_1 =0, \quad DB_2= 0, \quad DB_3 = 0, \quad DB_4= 0,\\
         & \Delta B_0= 0, \quad \Delta B_1= 0, \quad \Delta B_2 = 0, \quad \Delta B_3= 0.
         \end{aligned}
        \end{align}

        \item[ii. ] On $\scri$, 
        \begin{equation}
        \lambda \simeq 0, \qquad \mu\simeq 0, \qquad \Phi_{02'}\simeq 0, \qquad \phi_4\simeq0.
        \end{equation}

         \item[iii.] On $\NN_\star$,
        \begin{align*}
        \begin{aligned}
        & 2 \kappa_{1} (\Phi_{20} - \rho \lambda -  \mu \bar{\sigma}) + \bar{\sigma} \edt \kappa_{2} - 2 \kappa_{2} (\bar\Psi_{1} + \edt \bar{\sigma}) + 3 \lambda \bar{\eth} \kappa_{0}
        + 2 \kappa_{0} \bar{\eth} \lambda - 2 \bar{\eth} \bar{\eth} \kappa_{1}\bumpeq 0,\\
        &B_0\bumpeq B_1\bumpeq B_2 \bumpeq B_3\bumpeq 0.
        \end{aligned}
        \end{align*}
    \end{itemize}
    Then the Killing spinor candidate $\kappa_{AB}$ satisfying the wave equation \eqref{eq:KSWaveEq} is a Killing spinor on $\mathcal{V}$.
\end{proposition} 

The structure of this section is outlined as follows:
\begin{itemize}
    \item In Section \ref{sec:ConstructionOfKSCandidate}, we discuss the construction of a Killing spinor candidate and introduce its projection onto a spin dyad $\{o^A,\iota^A\}$.
    \item In Section \ref{sec:FirstDecompOfKSConditions}, we present the decomposition of the Killing spinor conditions onto the spin dyad $\{o^A,\iota^A\}$.
    \item In Sections \ref{sec:ConditionsForH}, \ref{sec:ConditionsForB} and \ref{sec:ReducingF} we rewrite the Killing spinor conditions in terms of the Killing spinor candidate $\kappa_{AB}$ and the free data of Lemma \ref{lem:freedata}.
    \item Finally, in Section \ref{sec:CCKSIDIntuition}, we provide an example of Proposition \ref{prop:CCKSID} by making some choices on the free data.
\end{itemize}

\subsection{Construction of Killing spinor candidates}\label{sec:ConstructionOfKSCandidate}

In this section we briefly discuss the solvability of the asymptotic characteristic initial value problem for the Killing spinor candidate equation \eqref{eq:KSWaveEq}. In particular, we want to identify the data that needs to be prescribed on $\NN_\star\cup \mathscr{I}^+$ so that \eqref{eq:KSWaveEq} has a unique solution in a neighbourhood of $\mathcal{S}_\star \equiv \NN_\star\cap \mathscr{I}^+$ in $D^-(\NN_\star\cup \mathscr{I}^+)$. 

\begin{remark}
{\em More generally, one is interested in solutions on \emph{narrow rectangles} along either $\NN_\star$ or $\mathscr{I}^+$, see Figure \ref{fig:narrowrectangle}.}
\end{remark}

The spinor $\kappa_{AB}$ can be expanded  in terms of the spin dyad $\{o^A,\iota^A\}$ as
\begin{align*}
\kappa_{AB} = \kappa_2 o_A o_B - 2\kappa_1 o_{(A} \iota_{B)} +
\kappa_0 \iota_A \iota_B.
\end{align*}
where
\begin{align*}
\kappa_0 \equiv \kappa_{AB} o^A o^B, \qquad \kappa_1 \equiv
\kappa_{AB} o^A \iota^B, \qquad \kappa_2\equiv \kappa_{AB}
\iota^A\iota^B.
\end{align*}
 The components of the Killing spinor $\ka_0,\ka_1,\ka_2$ have spin weights $s=1,0,-1$.

\medskip
Decomposing the wave equation \eqref{eq:KSWaveEq} using the spin dyad $\{o^A,\,\iota^A\}$ yields three equations for the components $\kappa_0$, $\kappa_1$, $\kappa_2$. Namely, one has that 
\begin{subequations}
\begin{align}
&\kappa_{2} (- \Xi \phi_{0} - 4 \rho \sigma)
 + (-5 \gamma -  \bar{\gamma} + 2 \mu) D \kappa_{0}
 + 4 \tau D \kappa_{1}
 + D \Delta \kappa_{0}
  - 2 \rho \Delta \kappa_{0}\nonumber \\
& + \Delta D \kappa_{0}
 + (6 \alpha -  \pi) \delta \kappa_{0}
 - 4 \rho \delta \kappa_{1}
 -  \delta \bar\delta \kappa_{0}
 + (-2 \bar{\alpha} -  \bar{\pi} + 4 \tau) \bar\delta \kappa_{0}
 - 4 \sigma \bar\delta \kappa_{1}\nonumber \\
& + \kappa_{1} (2 \Xi \phi_{1} - 2 \bar{\pi} \rho + 8 \alpha \sigma - 2 \pi \sigma + 2 D \tau - 2 \delta \rho - 2 \bar\delta \sigma)\nonumber \\
& + \kappa_{0} (- \Xi \phi_{2} + 8 \alpha \bar{\alpha}  - 2 \bar{\alpha} \pi + 2 \alpha \bar{\pi} + 4 \gamma \rho + 2 \mu \rho + 2 \lambda \sigma - 12 \alpha \tau - 2 D \gamma + 2 \delta \alpha - 2 \bar\delta \bar{\alpha} + 2 \bar\delta \tau)\nonumber \\
& -  \bar\delta \delta \kappa_{0} = 0,\label{eq:WaveEqks0}\\
\nr \\
&\kappa_{1} (2 \Xi \phi_{2} + 4 \mu \rho + 4 \lambda \sigma - 4 \pi \tau)
 + (- \gamma -  \bar{\gamma} + 2 \mu) D \kappa_{1}
 + 2 \tau D \kappa_{2}
 + D \Delta \kappa_{1}
 - 2 \pi \Delta \kappa_{0}\nonumber \\
& - 2 \rho \Delta \kappa_{1}
 + \Delta D \kappa_{1}
 + 2 \lambda \delta \kappa_{0}
 + (2 \alpha -  \pi) \delta \kappa_{1}
 - 2 \rho \delta \kappa_{2}
 -  \delta \bar\delta \kappa_{1}
 + 2 \mu \bar\delta \kappa_{0}\nonumber \\
& + (2 \bar{\alpha} -  \bar{\pi}) \bar\delta \kappa_{1}
 - 2 \sigma \bar\delta \kappa_{2}\nonumber \\
& + \kappa_{0} (- \Xi \phi_{3} - 4 \alpha \mu + 3 \gamma \pi + \bar{\gamma} \pi -  \mu \pi + \lambda \bar{\pi} - 2 \lambda \tau -  \Delta \pi + \delta \lambda + \bar\delta \mu)\nonumber \\
& + \kappa_{2} (- \Xi \phi_{1} + 4 \bar{\alpha} \rho -  \bar{\pi} \rho -  \pi \sigma - 4 \rho \tau + D \tau -  \delta \rho -  \bar\delta \sigma)
 -  \bar\delta \delta \kappa_{1} = 0,\label{eq:WaveEqks1}\\
\nr\\
&\kappa_{0} (- \Xi \phi_{4} - 4 \lambda \mu)
 + (3 \gamma -  \bar{\gamma} + 2 \mu) D \kappa_{2}
 + D \Delta \kappa_{2}
 - 4 \pi \Delta \kappa_{1}
  - 2 \rho \Delta \kappa_{2}\nonumber \\
& + \Delta D \kappa_{2}
 + 4 \lambda \delta \kappa_{1}
 + (-2 \alpha -  \pi) \delta \kappa_{2}
 -  \delta \bar\delta \kappa_{2}
 + 4 \mu \bar\delta \kappa_{1}
 + (6 \bar{\alpha} -  \bar{\pi} - 4 \tau) \bar\delta \kappa_{2}\nonumber \\
& + \kappa_{1} (2 \Xi \phi_{3} - 8 \bar{\alpha} \lambda - 2 \gamma \pi + 2 \bar{\gamma} \pi - 2 \mu \pi + 2 \lambda \bar{\pi} + 4 \lambda \tau - 2 \Delta \pi + 2 \delta \lambda + 2 \bar\delta \mu)\nonumber \\
& + \kappa_{2} (- \Xi \phi_{2} + 8 \alpha \bar{\alpha}  + 2 \bar{\alpha} \pi - 2 \alpha \bar{\pi} - 4 \gamma \rho + 2 \mu \rho + 2 \lambda \sigma - 4 \alpha \tau - 4 \pi \tau + 2 D \gamma  - 2 \delta \alpha + 2 \bar\delta \bar{\alpha} - 2 \bar\delta \tau)\nonumber \\
& -  \bar\delta \delta \kappa_{2} = 0.\label{eq:WaveEqks2}
\end{align}
\end{subequations}

It is well known that in the characteristic problem for wave equations only the value of the field at the initial hypersurface needs to be prescribed. Any transversal derivatives can be computed from the constraints implied by the wave equation on the characteristic hypersurface. In the case of the Killing spinor candidate equation \eqref{eq:KSWaveEq} one has the following:

\begin{proposition}\label{prop:ODEforKS}
    Given values of the scalar fields, $\ka_0$, $\ka_1$ and $\ka_2$ on $\NN_\star\cup\scri$,  the system of equations \eqref{eq:WaveEqks0}-\eqref{eq:WaveEqks2} can be cast as a system of ordinary differential equations for:
    \begin{itemize}
    \item[(i)] The derivatives $\Delta\ka_0$, $\Delta\ka_1$, $\Delta\ka_2$ along $\NN_\star$,
        \item[(ii)] The derivatives $D\ka_0$, $D\ka_1$ and $D\ka_2$ along $\scri$.
    \end{itemize}
\end{proposition}

\begin{proof} The proof follows by inspection  using the NP commutators on the equations \eqref{eq:WaveEqks0}-\eqref{eq:WaveEqks2} on $\NN_\star$ and $\scri$.
\end{proof}

In addition to the above, we have the following proposition applying the existence result from \cite{Luk12,HilValZha20b}:

\begin{proposition}
    Given smooth scalar fields $\ka_0,\ka_1,\ka_2$ on $\mathcal{N}_\star\cup\mathscr{I}^+$ there exists a narrow rectangle along $\mathscr{I}^+$ in $D^-(\mathcal{N}_\star\cup\mathscr{I}^+)$ as in Figure \ref{fig:narrowrectangle} such that the wave equation \eqref{eq:KSWaveEq} has a unique solution. 
\end{proposition}

\begin{remark}
{\em Following the discussion in \ref{Section:AlgebraicKillingSpinorCandidate} and in particular Remark \ref{Remark:Data}, one can use formula \eqref{Formula:KSCandidate} to obtain values for the components $\kappa_0$, $\kappa_1$ and $\kappa_2$ on, say, $\mathcal{N}_\star$. The values on $\scri$ are fixed by the asymptotic values of the Killing vector.}
\end{remark}

\subsection{A first decomposition of the Killing spinor conditions}\label{sec:FirstDecompOfKSConditions}

Proposition \ref{Proposition:BasicKillingSpinorData} has identified a basic set of conditions on $\NN_\star\cup\mathscr{I}^+$ ensuring the existence of a Killing spinor in the development of the initial hypersurface. In the following sections, we investigate the interconnection between the various conditions with the aim of obtaining a \emph{reduced set of conditions}. 

\medskip
Using the spin dyad decomposition outlined in Appendix \ref{app:NPgauge}, we investigate the quantities $H_{AA'BC}$, $B_{ABCD}$ and $F_{A'BCD}$. A direct computation shows that the condition 
\begin{equation*}
    H_{AA'BC}=0
\end{equation*}
can be decomposed as
\begin{subequations}
    \begin{align}
&H_{00} \equiv  D \kappa_{0} = 0,\label{eq:HDecomposition00}\\
&H_{10} \equiv  2 D \kappa_{1}
 + \bar{\eth} \kappa_{0}-2 \kappa_{0} \pi
 + 2 \kappa_{1} \rho
  = 0,\label{eq:HDecomposition10}\\
&H_{20} \equiv  D \kappa_{2}
 + 2 \bar{\eth} \kappa_{1}-2 \kappa_{0} \lambda
 - 2 \kappa_{1} \pi
 + 2 \kappa_{2} \rho
 = 0,\label{eq:HDecomposition20}\\
&H_{30} \equiv  \bar{\eth} \kappa_{2}- 2 \kappa_{1} \lambda  = 0,\label{eq:HDecomposition30}\\
&H_{01} \equiv   \edt \kappa_{0}+ 2 \kappa_{1} \sigma= 0,\label{eq:HDecomposition01}\\
&H_{11} \equiv   \Delta \kappa_{0}
 + 2 \edt \kappa_{1}+ 2 \kappa_{1} (\bar{\alpha} + \beta)
 - 2 \kappa_{0} \gamma
 - 2 \kappa_{0} \mu
 + 2 \kappa_{2} \sigma
 = 0,\label{eq:HDecomposition11}\\
&H_{21} \equiv 2 \Delta \kappa_{1}
 + \edt \kappa_{2} + 2 \kappa_{2} (\bar{\alpha} + \beta)
 - 2 \kappa_{1} \mu
 = 0,\label{eq:HDecomposition21}\\
&H_{31} \equiv   \Delta \kappa_{2} +2\kappa_{2} \gamma
 = 0.\label{eq:HDecomposition31}
\end{align}
\end{subequations}
The \emph{Buchdahl constraint} 
\begin{equation*}
    B_{ABCD}=0
\end{equation*}
is equivalent to the conditions
\begin{subequations}
    \begin{align}
&B_0\equiv \kappa_{1} \phi_{0}
 -  \kappa_{0} \phi_{1} = 0,\label{eq:BDecomp1}\\
&B_1 \equiv \frac{1}{4} \kappa_{2} \phi_{0}
 + \frac{1}{2} \kappa_{1} \phi_{1}
 -  \frac{3}{4} \kappa_{0} \phi_{2} = 0,\\
&B_2\equiv  \frac{1}{2} \kappa_{2} \phi_{1}
 -  \frac{1}{2} \kappa_{0} \phi_{3} = 0,\\
&B_3\equiv \frac{3}{4} \kappa_{2} \phi_{2}
 -  \frac{1}{2} \kappa_{1} \phi_{3}
 -  \frac{1}{4} \kappa_{0} \phi_{4} = 0,\\
&B_4 \equiv \kappa_{2} \phi_{3}
 -  \kappa_{1} \phi_{4} = 0.\label{eq:BDecomp2}
\end{align}
\end{subequations}
Recall here Assumption \ref{assump:BonNstar} in that we assume that these equations hold on $\NN_\star$.

\medskip
Finally, the equation 
\begin{equation*}
    F_{A'BCD}=0,
\end{equation*}
by definition, is equivalent to 
\begin{equation*}
    \nabla^Q{}_{A'}B_{QBCD}=0.
\end{equation*}
This equation is formally identical to the spin-2 massless equation. This observation allows us to obtain a reduction on the number of Killing spinor conditions. In the following sections, we will analyse each of these conditions to prove Proposition \ref{prop:CCKSID}.

\subsection{Deriving conditions on $\NN_\star\cup\scri$ equivalent to $H_{AA'BC}=0$}\label{sec:ConditionsForH}

In this section, we analyse the equations \eqref{eq:HDecomposition00}-\eqref{eq:HDecomposition31} to prove the following Lemma:

\begin{lemma}[Conditions on $\NN_\star\cup\scri$ equivalent to $H_{AA'BC}=0$]
\label{lem:ConditionsForH}
    Assume that the conformal gauge conditions of Proposition \ref{Lemma:ConformalGauge} are satisfied, Assumption \ref{Assumption:KSonScri} holds and equations \eqref{eq:HDecomposition00}, \eqref{eq:HDecomposition10} and \eqref{eq:HDecomposition20} hold along $\NN_\star$. If the following holds:
    \begin{itemize}
        \item[i. ] on $\scri$, 
        \begin{equation}
        \lambda \simeq 0, \qquad \mu\simeq 0, \qquad \Phi_{02'}\simeq 0;
        \end{equation}
        
        \item[ii. ] on $\mathcal{S}_\star$, $\kappa_0$ satisfies the constraint equation
        \begin{equation}
        \edt \kappa_{0} -\frac{2}{3} \sigma = 0
        \end{equation}
        and $\kappa_0$ is propagated along $\scri$ according to the transport equation
        \begin{equation}
        \Delta \kappa_{0} - \frac{2}{3} \tau- 2 \kappa_{0} \gamma\simeq 0;
        \end{equation}

         \item[iii.] on $\NN_\star$,
        \begin{align*}
        \begin{aligned}
        & 2 \kappa_{1} (\Phi_{20} - \rho \lambda -  \mu \bar{\sigma}) + \bar{\sigma} \edt \kappa_{2} - 2 \kappa_{2} (\bar\Psi_{1} + \edt \bar{\sigma}) + 3 \lambda \bar{\eth} \kappa_{0}
        + 2 \kappa_{0} \bar{\eth} \lambda - 2 \bar{\eth} \bar{\eth} \kappa_{1}\bumpeq 0,\\
        &B_0\bumpeq B_1\bumpeq B_2 \bumpeq B_3\bumpeq 0;
        \end{aligned}
        \end{align*}
    \end{itemize}

then 
\[
H_{AA'BC}=0 \qquad on \quad \NN_\star\cup\scri.
\]
\end{lemma}

\begin{remark}
    {\em  It is of interest to observe that under the assumptions of Lemma \ref{lem:ConditionsForH}, the components of the conformal Killing vector $\xi_{AA'}$ on $\scri$ are given by:
    \begin{subequations}
    \begin{eqnarray}
     && \xi_{11'}\simeq \Delta \kappa_1 \simeq 0, \label{CKVScri11}\\
     && \xi_{10'}\simeq D\kappa_2 +\frac{2}{3}\pi \simeq 0, \label{CKVScri10}\\
     && \xi_{01'}\simeq -\Delta \kappa_0 + 2\gamma \kappa_0 + \frac{2}{3}\tau \simeq 0, \label{CKVScri01}\\
     && \xi_{00'} \simeq -D\kappa_1 + \bar\eth\kappa_0 +\kappa_0\pi -\frac{2}{3}\rho  +1 \simeq 1. \label{CKVScri00}
    \end{eqnarray}
\end{subequations}
The relation \eqref{CKVScri11} is automatically satisfied by the requirement $\kappa_1\simeq -\frac{1}{3}$. Relation \eqref{CKVScri10} is the same as \eqref{eq:HDecomposition20NI} while \eqref{CKVScri01} is the transport equation for  $\kappa_0$ ---namely, equation \eqref{eq:HDecomposition11NI}. Finally, equations \eqref{CKVScri00} and \eqref{eq:HDecomposition10NI} satisfy the additional constraint 
\[
\bar \eth \kappa_0 \simeq \frac{2}{3}\rho.
\]
}
\end{remark}

\ni To prove Lemma \ref{lem:ConditionsForH}, we will analyse the equation $H_{AA'BC}=0$ on $\scri$ in Section \ref{sec:HonScri} and on $\NN_\star$ in Section \ref{sec:HonNstar}.

\subsubsection{Analysis of $H_{AA'BC}=0$ on $\scri$}\label{sec:HonScri}

In this section, we investigate the conditions that give rise to 
\[
H_{AA'BC}\simeq 0.
\]
Assumption \ref{Assumption:KSonScri} together with the decomposition of the Killing spinor equation as given by equations \eqref{eq:HDecomposition00}-\eqref{eq:HDecomposition31} implies for the condition $H_{AA'BC}\simeq0$ that
\begin{subequations}
    \begin{align}
& D \kappa_{0} \simeq 0,\label{eq:HDecomposition00NI}\\
&2 D \kappa_{1}
 + \bar{\eth} \kappa_{0}-2 \kappa_{0} \pi
 - \frac{2}{3} \rho
  \simeq  0,\label{eq:HDecomposition10NI}\\
&  D \kappa_{2}
 + \frac{2}{3} \pi
 \simeq 0,\label{eq:HDecomposition20NI}\\
&\lambda  \simeq 0,\label{eq:HDecomposition30NI}\\
& \edt \kappa_{0} -\frac{2}{3} \sigma\simeq  0,\label{eq:HDecomposition01NI}\\
& \Delta \kappa_{0}
 - \frac{2}{3} \tau
 - 2 \kappa_{0} \gamma
 \simeq 0,\label{eq:HDecomposition11NI}\\
&  \mu \simeq  0. \label{eq:HDecomposition21NI}
\end{align}
\end{subequations}

\begin{remark}
{\em Observe that the conditions in Assumption \ref{Assumption:KSonScri} imply $\Delta \kappa_2 \simeq \eth \kappa_2 \simeq \bar\eth \kappa_2\simeq 0$ and $\Delta \kappa_1 \simeq \eth \kappa_1 \simeq \bar\eth \kappa_1\simeq 0$. However, they do not specify the transverse derivatives $D\kappa_2$ and $D\kappa_1$. }    
\end{remark}

Given a fixed background geometry, a natural interpretation of the various components of $H_{AA'BC}\simeq 0$ is as follows: 
\begin{itemize}
\item[(i)] equations \eqref{eq:HDecomposition30NI} and \eqref{eq:HDecomposition21NI} are restrictions on the background geometry. 
\item[(ii)] equation \eqref{eq:HDecomposition01NI} is a \emph{constraint equation} for the components $\kappa_0$ which can be solved at, say, $\mathcal{S}_\star$.
\item[(iii)] Equation \eqref{eq:HDecomposition11NI} is a \emph{transport equation} for the component of $\kappa_0$, $\kappa_1$ and $\kappa_2$ along the null generators of $\scri$. Given initial data for $\kappa_0$, $\kappa_1$ and $\kappa_2$ on $\mathcal{S}_\star$, these equations determine their value along $\scri$. Naturally, one would like to make use of the initial data obtained from solving the constraint  \eqref{eq:HDecomposition01NI}.
\item[(iv)] Finally, once the values of the components of $\kappa_{AB}$ on $\scri$ are known, equations \eqref{eq:HDecomposition00NI}-\eqref{eq:HDecomposition20NI} become \emph{algebraic conditions} which allow to determine the transverse derivatives $D\kappa_0$, $D\kappa_1$ and $D\kappa_2$. 
\end{itemize}

In order for the above scheme to be consistent one needs to verify a number of solvability and consistency conditions. More precisely, one needs to verify that:

\begin{itemize}
    \item[(a)] the \emph{constraint} \eqref{eq:HDecomposition01NI} can be solved for a given choice of $\sigma$. One also needs to examine the conditions for uniqueness of the solution.

    \item[(b)] One needs to understand the conditions upon which equation \eqref{eq:HDecomposition01NI} is propagated by virtue of the transport equation \eqref{eq:HDecomposition11NI}. In other words, one would like to solve the constraints only at $\mathcal{S}_\star$.

    \item[(c)] One needs to understand to what extent the value of the transverse derivatives $D\kappa_0$, $D\kappa_1$ and $D\kappa_2$ as given by equations \eqref{eq:HDecomposition00NI}-\eqref{eq:HDecomposition20NI} is consistent with the values given by the Killing spinor candidate equation \eqref{eq:KSWaveEq} ---see Proposition \ref{prop:ODEforKS}. 
\end{itemize}

In the remaining of this section we ellaborate on the points (a)-(c) above. 

\medskip
Regarding point (a) it is observed that the formal adjoint of $\eth$ acting on an object of spin-weight $1$ is the operator $\bar\eth$ acting on operators of spin-weight $2$. The latter has trivial Kernel ---see e.g. \cite{Ste91}, Lemma 3.9.4. Thus, making use of the Fredholm alternative it follows that equation \eqref{eq:HDecomposition01NI} can be solved for a given choice of the spin-connection coefficient $\sigma$. This solution, however, is not unique as the Kernel of $\eth$ acting on objects of spin-weight $1$ is non-trivial ---it is spanned by the spin-weighted harmonics ${}_1 Y_{1m}$. 

\medskip
For point (b) above a calculation shows that 
the left hand side of \eqref{eq:HDecomposition01NI} satisfies the equation
\[
    \Delta(\edt \kappa_{0}- \frac{2}{3} \sigma)\simeq ( 3 \gamma -  \bar{\gamma} )(\edt \kappa_{0}-\frac{2}{3}\si)- \frac{2}{3} \Phi_{02'}-2 \kappa_{0} \Phi_{12'}.
\]
The above equation ensures the propagation of the constraint if and only if the non-homogeneous term vanishes. That is, we require
\begin{equation}
\Phi_{02'}+3 \kappa_{0} \Phi_{12'}\simeq 0.
\label{Alignment:Ricci}
\end{equation}
This is an alignment condition between the tracefree Ricci tensor and the Killing spinor. 

\medskip
Finally, assume that $\lambda\simeq 0$ and $\mu\simeq 0$. Then the Ricci identities imply that 
\[
\Phi_{12'}\simeq 0.
\]
Thus, the alignment condition \eqref{Alignment:Ricci} takes the form
\[
\Phi_{02'} \simeq 0.
\]

\medskip
Finally, it remains to consider the equations for the transversal derivatives \eqref{eq:HDecomposition00NI}-\eqref{eq:HDecomposition20NI}.  Substituting these into the transport equations along $\scri$ implied by the Killing spinor candidate equation \eqref{eq:KSWaveEq} one finds, after some lengthy computations that these are trivially satisfied. Thus, the expressions for $D\kappa_0$, $D\kappa_1$ and $D\kappa_2$ obtained from solving the algebraic equations \eqref{eq:HDecomposition00NI}-\eqref{eq:HDecomposition20NI} are equivalent to that of solving the transport equations. This concludes our analysis of $H_{AA'BC}=0$ on $\scri$.

\subsubsection{Analysis of $H_{AA'BC}=0$ On $\NN_\star$}\label{sec:HonNstar}

An analysis analogous to the one on $\mathscr{I}^+$ follows, \emph{mutatis mutandi}, for the condition $H_{AA'BC}=0$ on $\NN_\star$ ($H_{AA'BC}\bumpeq 0$ for short). This analysis is, however, less clean than the one on $\mathscr{I}^+$ as the null hypersurface $\NN_\star$ has less structure. The interpretation of the decomposition of $H_{AA'BC}\bumpeq 0$ is as follows:
\begin{itemize}
\item[(ii')] equations \eqref{eq:HDecomposition00}, \eqref{eq:HDecomposition10} and \eqref{eq:HDecomposition20} constitute transport equations along $\NN_\star$ propagating the values of $\kappa_0$, $\kappa_1$ and $\kappa_2$ on $\mathcal{S}_\star$ along $\NN_\star$.
\item[(iii')] as in the case of $\scri$, the value of $\kappa_0$ and $\kappa_2$ is obtained, for prescribed $\kappa_1$, from solving the \emph{constraint equations} \eqref{eq:HDecomposition30} and \eqref{eq:HDecomposition01}. 
\item[(iv')] Finally, once the value of $\kappa_0$, $\kappa_1$, $\kappa_2$ along $\NN_\star$ is known, equations \eqref{eq:HDecomposition11}-\eqref{eq:HDecomposition31} become \emph{algebraic relations} giving the value of the transverse derivatives $\Delta\kappa_0$, $\Delta\kappa_1$ and $\Delta\kappa_2$. 
\end{itemize}

The solvability of the constraints \eqref{eq:HDecomposition30}-\eqref{eq:HDecomposition01} has already been discussed in the previous subsection. On still needs to consider:
\begin{itemize}
    \item[(b')] the propagation of the constraints \eqref{eq:HDecomposition30}-\eqref{eq:HDecomposition01} along $\NN_\star$ by virtue of the transport equations \eqref{eq:HDecomposition00}-\eqref{eq:HDecomposition20}.
    \item[(c')] The consistency of between the value of the transverse derivatives provided by equations \eqref{eq:HDecomposition11}-\eqref{eq:HDecomposition31} and those given by the Killing candidate equation \eqref{eq:KSWaveEq}.
\end{itemize}

In order to address point (b') above, we consider the derivatives of the constraints along $\mathcal{N}_\star$. A computation yields the equations
\begin{align*}
\begin{aligned}
      D(\edt \kappa_{0}+ 2 \kappa_{1} \sigma)&\bumpeq  \rho (\edt \kappa_{0} + 2 \kappa_{1} \sigma )+2 \Xi( \kappa_{0} \phi_{1}-\kappa_{1} \phi_{0})\\
      &\bumpeq \rho (\edt \kappa_{0} + 2 \kappa_{1} \sigma ) + 2 \Xi B_0.
    \end{aligned}
\end{align*}
Thus, the above equation ensures propagation of the constraint $\edt \kappa_{0}+ 2 \kappa_{1} \sigma$ if the component $B_0$ of the Buchdahl constraint vanishes along $\mathcal{N}_\star$. The equation for the constraint $\bar{\eth} \kappa_{2}
 - 2 \kappa_{1} \lambda$ is more involved. In this case one obtains
\begin{align*}
\begin{aligned}
   D(\bar{\eth} \kappa_{2}
 - 2 \kappa_{1} \lambda) \bumpeq &- \rho( \bar{\eth} \kappa_{2}-2 \kappa_{1}\la)+2 \kappa_{1} (\Phi_{20} - \rho \lambda -  \mu \bar{\sigma}) + \bar{\sigma} \edt \kappa_{2} - 2 \kappa_{2} (\bar\Psi_{1} + \edt \bar{\sigma})\\& \hspace{1cm} + 3 \lambda \bar{\eth} \kappa_{0}
  + 2 \kappa_{0} \bar{\eth} \lambda - 2 \bar{\eth} \bar{\eth} \kappa_{1}.
 \end{aligned}
\end{align*}
where we have used the transport equations \eqref{eq:HDecomposition10} and \eqref{eq:HDecomposition20}.
In order to address point (c') we start by substituting \eqref{eq:HDecomposition31} into the transport equations implied by the Killing spinor candidate equation \eqref{eq:KSWaveEq}.  After applying the NP commutators, the transport equations \eqref{eq:HDecomposition21}, \eqref{eq:HDecomposition20}, \eqref{eq:HDecomposition30} and the NP Ricci equations one finds the compatibility condition
\[
    3 \kappa_{2} \phi_{2} - 2 \kappa_{1} \phi_{3} -  \kappa_{0} \phi_{4} \bumpeq 0.
\]
Following a similar procedure with equation \eqref{eq:HDecomposition21} one finds, after using the NP commutators, the transport equations \eqref{eq:HDecomposition11}, \eqref{eq:HDecomposition10}, \eqref{eq:HDecomposition20}, \eqref{eq:HDecomposition30} and the NP Ricci equations the compatibility condition
\[
    \kappa_{2} \phi_{1} -  \kappa_{0} \phi_{3} \bumpeq 0.
\]
Finally, for \eqref{eq:HDecomposition11} using the NP commutators, the transport equations \eqref{eq:HDecomposition00},  \eqref{eq:HDecomposition01},  \eqref{eq:HDecomposition10},  \eqref{eq:HDecomposition20} and the NP Ricci equations one concludes that 
\begin{equation}
    \kappa_{2} \phi_{0} + 2 \kappa_{1} \phi_{1} - 3 \kappa_{0} \phi_{2} \bumpeq 0.
\end{equation}

\begin{remark}
{\em These compatibility conditions are precisely $B_1\bumpeq 0$, $B_2\bumpeq 0$ and $B_3\bumpeq0$ respectively.}
\end{remark}

\ni This completes the proof of Lemma \ref{lem:ConditionsForH}.

\subsection{Deriving conditions on $\scri$ equivalent to $B_{ABCD}=0$}\label{sec:ConditionsForB}

In this section, we analyse the equations \eqref{eq:BDecomp1}-\eqref{eq:BDecomp2} to prove the following lemma:

\begin{lemma}[Conditions on $\scri$ equivalent to $B_{ABCD}=0$]\label{lem:ConditionsForB}
Assume that the conformal gauge conditions of Proposition \ref{Lemma:ConformalGauge} are satisfied and Assumption \ref{Assumption:KSonScri} holds. If on $\mathcal{S}_\star$, $\kappa_0$, $\phi_0$, $\phi_1$ and $\phi_2$ can be chosen to satisfy
\begin{align*}
 \begin{aligned}
\frac{1}{3}\phi_{0}
 +\kappa_{0} \phi_{1} =& 0, \\
 \frac{2}{3} \phi_{1}
 +3 \kappa_{0} \phi_{2} =& 0,\\
 \end{aligned}
 \end{align*}
 with $\phi_3=0$, and on $\scri$, $\phi_4\simeq 0$, then
 \[
 B_{ABCD}\simeq0
 \]
\end{lemma}

\begin{remark}
\emph{By Assumption \ref{assump:BonNstar}, we do not need to consider conditions for $B_{ABCD}\bumpeq 0$.} 
\end{remark}

\begin{proof}

Under Assumption \ref{Assumption:KSonScri}, the dyad componeonts of $B_{ABCD}$ \eqref{eq:BDecomp1}-\eqref{eq:BDecomp2} reduce to
\begin{subequations}
    \begin{align}
&B_0\equiv -\frac{1}{3}\phi_{0}
 -  \kappa_{0} \phi_{1} = 0, \label{eq:B0onScri}\\
&B_1 \equiv -\frac{2}{3} \phi_{1}
 -  3 \kappa_{0} \phi_{2} = 0,\label{eq:B1onScri}\\
&B_2\equiv - \kappa_{0} \phi_{3} = 0,\\
&B_3\equiv  \frac{2}{3}\phi_{3}
 -   \kappa_{0} \phi_{4} = 0,\\
&B_4 \equiv \frac{1}{3} \phi_{4} = 0.
\end{align}
\end{subequations}
Inspecting the components $B_2$ and $B_4$ readily implies that $\phi_4=\phi_3=0$ is equivalent to $B_2=B_3=B_4=0$ on $\scri$. Now, the fourth conformal field equation \eqref{eq:CFE4a}-\eqref{eq:CFE4b} are a set of transport equations along $\scri$ and $\NN_\star$ for the components of the rescaled Weyl spinor. Consider setting $\phi_4\simeq0$ on $\scri$, then the transport equations for the components of the Weyl spinor along $\scri$ take the form ---see the system of equations \eqref{eq:CFE4onScri}:
 \begin{align}
 \begin{aligned}\label{eq:CFE4onScriphi4To0}
 \Delta \phi_{0} - \delta \phi_{1}={}&-2 \phi_{1} \beta
 + 4 \phi_{0} \gamma
 -  \phi_{0} \mu
 + 3 \phi_{2} \sigma
 - 4 \phi_{1} \tau,\\
\Delta \phi_{1}- \delta \phi_{2}={}&2 \phi_{1} \gamma
 - 2 \phi_{1} \mu
 + 2 \phi_{3} \sigma
 - 3 \phi_{2} \tau,\\
\Delta \phi_{2}-\delta \phi_{3}={}&2 \phi_{3} \beta
 - 3 \phi_{2} \mu
 - 2 \phi_{3} \tau,\\
\Delta \phi_{3}={}&
 - 2 \phi_{3} \gamma
 - 4 \phi_{3} \mu
 \end{aligned}
 \end{align}
 The last of these equations is a homogeneous transport equation for $\phi_3$. Hence, setting $\phi_3=0$ on $\mathcal{S}_\star$ means that $\phi_3\simeq0$. It is important to note that the Weyl component $\phi_3$ does not constitute part of our free data on $\mathcal{S}_\star$. It is specified by the equation \eqref{eq:CFE3l} on $\mathcal{S}_\star$
\begin{equation*}
\delta \Phi_{11}
 - \bar\delta \Phi_{02}={}\bar{\phi}_{3} - 2 \Phi_{02} \alpha
 + 2 \Phi_{02} \bar{\beta},
\end{equation*}
 where $\Phi_{02}$ is part of the free data and the quantities $\alpha$, $\beta$ and $\phi_{11}$ are calculated from the specification of the differential operator $\delta$, the fact that $\tau=0$ on $\mathcal{S}_\star$ and the equation \eqref{eq:Riccij} on $\mathcal{S}_\star$
\begin{equation*}
\bar\delta \beta - \delta \alpha=
 + \Phi_{11}
 -  \alpha \bar{\alpha}
 + 2 \alpha \beta
 -  \beta \bar{\beta}.
\end{equation*}

  The first three equations of \eqref{eq:CFE4onScriphi4To0} with $\phi_3=0$ are
 \begin{align*}
 \begin{aligned}
 \Delta \phi_{0} - \delta \phi_{1}={}&-2 \phi_{1} \beta
 + 4 \phi_{0} \gamma
 -  \phi_{0} \mu
 + 3 \phi_{2} \sigma
 - 4 \phi_{1} \tau,\\
\Delta \phi_{1}- \delta \phi_{2}={}&2 \phi_{1} \gamma
 - 2 \phi_{1} \mu
 - 3 \phi_{2} \tau,\\
\Delta \phi_{2}-\delta \phi_{3}={}&
 - 3 \phi_{2} \mu.
 \end{aligned}
 \end{align*}
This is a closed system of equations for $\phi_0$, $\phi_1$ and $\phi_2$ and therefore, given values on $\mathcal{S}_*$, this system propagates $\phi_0$, $\phi_1$ and $\phi_2$ onto $\scri$. Moreover, the equations \eqref{eq:B0onScri} and \eqref{eq:B1onScri} are a pair of simultaneous equations on $\mathcal{S}_\star$. The components of the Weyl spinor $\phi_0$, $\phi_1$ and the real part of $\phi_2$ constitute parts of the free data of Lemma \ref{lem:freedata}. The complex part of $\phi_2$ is evaluated on $\mathcal{S}_\star$ from the third conformal field equation \eqref{eq:CFE3k}. Thus if $\phi_0$, $\phi_1$ and $\phi_2$ can be chosen to satisfy the equations 
\begin{align*}
 \begin{aligned}
\frac{1}{3}\phi_{0}
 +\kappa_{0} \phi_{1} =& 0, \\
 \frac{2}{3} \phi_{1}
 +3 \kappa_{0} \phi_{2} =& 0,\\
 \end{aligned}
 \end{align*}
 on $\mathcal{S}_\star$, $\phi_3=0$ on $\mathcal{S}_\star$ and $\phi_4=0$ on $\scri$, then
\[
B_{ABCD}\simeq 0.
\]
\end{proof}

\subsection{A reduction of the equation $F_{A'BCD}=0$}\label{sec:ReducingF}

Finally, in this section, we proof the following lemma:
\begin{lemma}[Conditions on $\NN_\star\cup\scri$ equivalent to $F_{A'BCD}=0$]\label{lem:ConditionsForF}
    Assume that the statement of Lemma \ref{lem:ConditionsForB} holds, that is $B_{ABCD}=0$ on $\NN_\star\cup\scri$. If on $\mathcal{S}_\star$ the equations
    \begin{align}
    \begin{aligned}
     & D B_1 =0, \quad DB_2= 0, \quad DB_3 = 0, \quad DB_4= 0,\\
     & \Delta B_0= 0, \quad \Delta B_1= 0, \quad \Delta B_2 = 0, \quad \Delta B_3= 0.
     \end{aligned}
    \end{align}
    hold, then 
    \[
    F_{A'BCD}=0 \qquad on \quad \NN_\star\cup\scri
    \]
\end{lemma}   

\begin{proof}
Recall that the equation 
\begin{equation*}
    F_{A'BCD}=0
\end{equation*}
is, by definition, given by
\begin{equation*}
    \nabla^Q{}_{A'}B_{QBCD}=0,
\end{equation*}
which, in turn, is formally identical to the spin-2 massless equation. Thus, we can obtain a reduction on the number of Killing spinor conditions. Expanding the massless spin-2 equation for $B_{ABCD}$ in terms of the components $B_0,\ldots,B_4$ on $\mathscr{I}^+$ and using that from Lemma \ref{lem:ConditionsForB}, $B_{ABCD}\simeq 0$  (and therefore also $\Delta B_{ABCD}\simeq 0$, $\eth B_{ABCD}\simeq 0$ and $\bar{\eth} B_{ABCD}\simeq 0$), we obtain that
\begin{equation}
D B_1\simeq 0, \quad DB_2\simeq 0, \quad DB_3 \simeq 0, \quad DB_4\simeq 0.
\label{ConditionFScriReduced}
\end{equation}
That is, on $\scri$, the equations \eqref{ConditionFScriReduced} are equivalent to $F_{A'BCD}\simeq0$. We note that there is no condition on $DB_0$. Recall that $B_{ABCD}$ also satisfies the wave equation \eqref{Eq:WaveEqForB}. Using that $B_{ABCD}\simeq 0$ on $\mathscr{I}^+$ and the NP commutators, the wave equation for $B_{ABCD}$ is trivially satisfied on $\scri$. There are no consistency conditions arising from combining the conditions \eqref{ConditionFScriReduced} and the the wave equation \eqref{Eq:WaveEqForB}. Thus, the equations \eqref{ConditionFScriReduced} hold $\scri$ if they hold on $\mathcal{S}_\star$. 

Analogously, on $\mathcal{N}_\star$, using that $B_{ABCD}\bumpeq0$ (and therefore  $D B_{ABCD}\bumpeq 0$, $\eth B_{ABCD}\bumpeq 0$ and $\bar{\eth} B_{ABCD}\bumpeq 0$), we obtain that
\begin{equation}\label{eq:ConditionsFNstarReduced}
\Delta B_0\bumpeq 0, \quad \Delta B_1\bumpeq 0, \quad \Delta B_2 \bumpeq 0, \quad \Delta B_3\bumpeq 0.
\end{equation} 
The component $\Delta B_4$ is unconstrained. Again, using that $B_{ABCD}\bumpeq 0$ on $\NN_\star$, the wave equation for $B_{ABCD}$ using the NP commutators is trivially satisfied and therefore can be propagated given that the equations \eqref{eq:ConditionsFNstarReduced} hold.
\end{proof}

Combining Lemmas \ref{lem:ConditionsForH}, \ref{lem:ConditionsForB} and \ref{lem:ConditionsForF} concludes the proof of Proposition \ref{prop:CCKSID}. In the next section, we provide some intuition for Proposition \ref{prop:CCKSID}.

\subsection{Some intuition: non-expanding horizons}\label{sec:CCKSIDIntuition}

In this subsection we provide some intuition into the conditions obtained in Propositions \ref{prop:CCKSID} by considering specific choices of the free specifiable data. It turns out that this choice has an analogue to the case where the null hypersurfaces are non-expanding horizons\\

As discussed in Proposition \ref{lem:freedata}, the component $\phi_0$ of the rescaled Weyl tensor encodes the freely specifiable data on $\NN_\star$. Accordingly, a particular interesting case to obtain intuition is to set $\phi_0\bumpeq 0$.

\medskip
If $\phi_0\bumpeq 0$, then $\kappa_0 \phi_1\bumpeq 0$. \emph{Choosing $\phi_1\bumpeq 0$ implies that $\ka_0\bumpeq 0$ and vice versa.} To see this, assume $\phi_1\bumpeq 0$ and notice that with $\phi_0\bumpeq 0$, $B_1\bumpeq 0$ implies that $\ka_0\phi_2\bumpeq 0$ and $B_2 \bumpeq 0$ implies that $\ka_0\phi_3 \bumpeq 0$. If $\ka_0\neq  0$ then $\phi_2=\phi_3\bumpeq 0$ and $B_3\bumpeq 0$ implies that $\phi_4\bumpeq 0$. If $\ka_0\bumpeq 0$ this does not constraint $\phi_2$ or $\phi_3$. On the other hand, if we assume $\kappa_0\bumpeq 0$ then with $\phi_0\bumpeq 0$, the condition $B_1\bumpeq 0$ implies that $\ka_1\phi_1\bumpeq 0$. If we assume $\ka_1 \bumpeq 0$ then condition $B_2\bumpeq 0$ implies that $\ka_2 \bumpeq 0$. Since we assume that the Killing spinor does not vanish, we choose $\phi_1 \bumpeq 0$. 

\smallskip
Now, substituting the conditions $\phi_1\bumpeq 0$ and $\ka_0\bumpeq 0$ into the condition in Proposition \ref{prop:CCKSID} one obtains 
\begin{equation*}
    2 \kappa_{1} (\Phi_{20} - \rho \lambda)- 2 \bar{\eth} \bar{\eth} \kappa_{1}\bumpeq 0.
\end{equation*}
Recall that under the current assumptions the conformal field equations for the tracefree Ricci tensor imply the transport equation 
\[
    D \Phi_{02}-\Phi_{02} \rho\bumpeq 0.
\]
Moreover, we also have the following transport equation for $\rho$:
\[
 D\rho-\rho^2\bumpeq 0.
\]
From Lemma \ref{Lemma:ConformalGauge} we have that $\rho=0$ on $\mcS_\star$. Thus, using the last transport equation one can conclude that $\rho\bumpeq 0$. The latter, in turn, implies that $\Phi_{02'}\bumpeq 0$. Accordingly, the differential condition on $\NN_\star$ in Proposition \ref{prop:CCKSID} reduces to  
\begin{equation*}
 \bar{\eth} \bar{\eth} \kappa_{1}\bumpeq 0.
\end{equation*}

\begin{remark}
    {\em This condition has an analogue in the analysis of \cite{ColValRac18}. Observe that in that case, the null hypersurface was assumed to be a non-expanding horizon.}
\end{remark}

\section{The asymptotic characterisation of Kerr}\label{sec:AsympKerr}

In this section we conclude our analysis by combining the characteristic conformal Killing spinor initial data, and develop the asymptotic characterisation of Kerr using Theorem \ref{Theorem:SpinorialCharacterisationKerr}. In section \ref{sec:confrepKerr}, we begin by looking at a simple conformal representation of the Kerr spacetime. This will allow us to conclude the analysis started in Section \ref{Section:KillingSpinorsGeneral} ---see, in particular, Theorem \ref{Theorem:SpinorialCharacterisationKerr} and Lemma \ref{Lemma:AsymptoticSimplicity}. 

\subsection{A conformal representation of the Kerr spacetime}\label{sec:confrepKerr} 

In the following, let $(\tilde{\mathcal{M}}_{\mathrm{Kerr}},\tilde{\bmg}_{\mathrm{Kerr}})$ denote portion of a member of the Kerr family of spacetimes described by the standard \emph{Boyer-Lindquist coordinates} $(t,r,\theta,\varphi)$ ---see equation \eqref{eq:KerrBL}. There exists a spin dyad $\{ \tilde{o}^A,\tilde{\iota}^A \}$ aligned with the principal directions of the Weyl spinor such that 
\begin{eqnarray*}
&& \Psi_{ABCD} = -\frac{m}{6\tilde\varkappa^3} \tilde{o}_{(A}\tilde{o}_{B}\tilde{\iota}_C\tilde{\iota}_{D)}, \\
&& \tilde\kappa_{AB}=\frac{2}{3}\tilde\varkappa \tilde{o}_{(A}\tilde{\iota}_{B)},
\end{eqnarray*}
where, for convenience, we have set
\[
\tilde\varkappa \equiv r-\mathrm{i}a \cos\theta,
\]
and $m$, $a$ denote, respectively, the mass and angular parameters of the Kerr family ---see. e.g. \cite{AndBaeBlu16}. 

\medskip
A simple conformal rescaling giving rise to a portion of $\mathscr{I}^+$ is obtained by introducing the conformal factor 
\[
\Omega \equiv \frac{1}{r}. 
\]
Defining an \emph{unphysical spin dyad} $\{o^A,\iota^A \}$ via the conditions
\[
o_A= \tilde{o}_A, \quad o^A=\Omega^{-1}\tilde{o}^A, \quad \iota_A =\Omega \tilde{\iota}_A, \quad \iota^A= \tilde{\iota}^A,
\]
one readily finds that 
\begin{eqnarray*}
&& \phi_{ABCD}= - \frac{m}{6 \varkappa^3}\iota_{(A}\iota_B o_C o_{D)},\\
&& \kappa_{AB}= \frac{2}{3} \varkappa o_{(A}\iota_{B)},
\end{eqnarray*}
where
\[
\varkappa \equiv 1-\mathrm{i}a \Omega \cos\theta.
\]
Observe that $\varkappa \simeq 1$. 

\begin{remark}
{\em It is important to stress that the dyad $\{o^A,\,\iota^A \}$ is different to the one used in the construction of the \emph{Stewart gauge} in Section \ref{Section:StewartGauge}.}  
\end{remark}

It is now observed that 
\[
o_{(A}o_B\iota_C\iota_{D)}o^C\iota^D= -\frac{1}{3}o_{(A}\iota_{B)}, \qquad o_{(A}\iota_{B)} o^A\iota^B =-\frac{1}{2}, \qquad o_{(A}o_B\iota_C\iota_{D)}o^A\iota^B o^C\iota^D=\frac{1}{6}. 
\]
It follows from the above identities that
\begin{eqnarray*}
&& \mathcal{H}^2 = -\frac{m^2}{324 \varkappa^4 }, \\
&& \chi = -\frac{m}{18 \varkappa}. 
\end{eqnarray*}
Given that 
\begin{equation}
    \mathcal{H}^2 = -\mathfrak{l}\chi^4,
    \label{UnphysicalErnstPotentialCondition}
\end{equation}
it follows then that
\[
\mathfrak{l}= \frac{648}{m^2}>0.
\]

\begin{remark}
    {\em The main take away from the previous computation is that $\mathfrak{l}$ is a real positive number.}
\end{remark}

\begin{remark}
    {\em A quick computation readily shows that condition \eqref{UnphysicalErnstPotentialCondition} is conformally invariant ---that is, it is independent of the particular choice of conformal factor used to construct the conformal extension in a neighbourhood of the future null infinity of the Kerr spacetime.}
\end{remark}

The previous discussion can be summarised in the following:

\begin{proposition}\label{prop:Kerr}
    Let $(\tilde{\mathcal{M}},\tilde{\bmg})$ satisfy  Assumption \ref{Assumption:AsymptoticFlatness} and the assumptions of Theorem \ref{Theorem:SpinorialCharacterisationKerr} with $\kappa_{AB}=O(1)$ near $\mathscr{I}^+$. Then one has that 
\begin{equation}
\label{eq:EqsForl}
H= 6/\chi, \qquad \mathcal{H}^2 = -\mathfrak{l} \chi^4.
\end{equation}
Moreover, $(\tilde{\mathcal{M}},\tilde{\bmg})$ is locally isometric to a portion of the asymptotic region of a member of the Kerr family of spacetimes with non-vanishing mass if and only if $\mathfrak{l}$ is a real positive number. 
\end{proposition}

\subsection{Characterising the Kerr spacetime in Stewart's gauge}\label{sec:freedataKerr}

In this section, we calculate $H$, $\mathcal{H}$ and $\chi$ and derive conditions on the free data of Lemma \ref{lem:freedata} so that equation \eqref{eq:EqsForl} of Proposition \ref{prop:Kerr} holds and that ensure $\mathfrak{l}$ is a real positive number. To this end, we prove the following lemma:

\begin{lemma}\label{lem:Kerrstewart}
    Let $(\tilde{\mathcal{M}},\tilde{\bmg})$ satisfy  Assumption \ref{Assumption:AsymptoticFlatness} and the assumptions of Theorem \ref{Theorem:SpinorialCharacterisationKerr} with $\kappa_{AB}=O(1)$ near $\mathscr{I}^+$. Then $(\tilde{\mathcal{M}},\tilde{\bmg})$ is locally isometric to a portion of the asymptotic region of a member of the Kerr family of spacetimes with non-vanishing mass if and only if 
    \begin{equation}
    \label{eq:Kerrcond}
    \mathrm{Im}(\edt\edt \bar{\sigma})=\mathrm{Im}(\Phi_{02} \bar{\sigma}) \qquad \text{on } \mathcal{S}_\star. 
    \end{equation}
\end{lemma}

\begin{remark}
{\em The GHP weight of the right-hand side of equation \eqref{eq:Kerrcond} is a weighted scalar of type $\{1,1\}$, see Chapter 3.8.2 of \cite{ODo03}. Thus it has spin weight $s=0$. Then the equation 
\begin{equation*}
\edt\edt \bar{\sigma}=\Phi_{02} \bar{\sigma}
\end{equation*}
has a unique solution for $\edt \bar{\sigma}$ on $\mathcal{S}_\star$ if and only if the $j=0$ part of $\Phi_{02} \bar{\sigma}$ vanishes, see the discussion below the table (4.15.60) in \cite{PenRin84}.}

\end{remark}

\begin{proof}
 We calculate $\mathcal{H}^2$ and $\chi$ on the initial cut $\mathcal{S}_\star$. In Stewart's gauge, we find that 
\begin{equation}\label{eq:SpecifyKerrChi}
    \mathcal{H}^2 = -4\phi_2^2, \qquad \chi = 2\phi_2,
\end{equation}
where we have used Assumption \ref{Assumption:KSonScri} and its implications. Thus, $\mathfrak{l}$ will be real and positive if $\phi_2^2$ is real and positive. Now, we investigate the conditions we may place on the free data of Lemma \ref{lem:freedata}. The condition that $\phi_2^2$ be real is equivalent to the condition 
\begin{equation*}
(\phi_2+\bar{\phi_2})(\phi_2-\bar{\phi_2})=0
\end{equation*}
We cannot simply set $(\phi_2+\bar{\phi_2})=0$ as this would mean that $\phi_2^2$ would be negative. Thus, we consider the condition $(\phi_2-\bar{\phi_2})=0$. We can write this in terms of the free data by considering the equation \eqref{eq:CFE3k}. Evaluating equation \eqref{eq:CFE3k} on $\mathcal{S}_\star$ yields
\begin{equation*}
\phi_{2} -  \bar{\phi}_{2}= 2 \Phi_{01} \alpha
 - 2 \Phi_{10} \bar{\alpha}
 + \Phi_{20} \sigma
 -  \Phi_{02} \bar{\sigma}
 + \delta \Phi_{10}
 -  \bar\delta \Phi_{01}
\end{equation*}
Then, by definition we have that
\begin{align*}
\delta \Phi_{10}={}&2 \Phi_{10'} \bar{\alpha}
 + \edt \Phi_{10},\\
\bar\delta \Phi_{01}={}&2 \Phi_{01} \alpha
 + \bar{\edt} \Phi_{01},
\end{align*}
together with the Ricci equation \eqref{eq:Riccieqq}
\begin{align*}
\begin{aligned}
\bar{\edt} \sigma ={}&\Phi_{01},\\
\edt \bar{\sigma}={}&\Phi_{10}.
\end{aligned}
\end{align*}
implies that 
\[
\phi_{2} -  \bar{\phi}_{2}= 
  \Phi_{20} \sigma
 -  \Phi_{02} \bar{\sigma}
 + \edt\edt \bar{\sigma}
 -  \bar\edt\bar{\edt} \sigma.  
\]
Thus, 
\begin{align*}
\begin{aligned}
2\mathrm{Im}(\phi_2)=&\phi_{2} -  \bar{\phi}_{2}\\
=& \Phi_{20} \sigma
 -  \Phi_{02} \bar{\sigma}
 + \edt\edt \bar{\sigma}
 -  \bar\edt\bar{\edt} \sigma\\
 =&2\mathrm{Im}\left(\edt\edt \bar{\sigma} -  \Phi_{02} \bar{\sigma}\right)
\end{aligned}
\end{align*}
where we note that the quantities in the final line constitute part of the free data on $\mathcal{S}_\star$ from Lemma \ref{lem:freedata}. Therefore, ensuring that $\sigma$ and $\Phi_{02}$ satisfy 
\[
\mathrm{Im}(\edt\edt \bar{\sigma})=\mathrm{Im}(\Phi_{02} \bar{\sigma})
\]
on $\mathcal{S}_\star$ ensures that $\phi_2^2$ is real and positive. 

\medskip
For $H$, we investigate the alignment condition \eqref{eq:unphysicalalignmentcondition}. Using Assumption \ref{Assumption:KSonScri}, the components of this equation are 
\begin{align*}
\phi_{0}={}&\frac{3}{2} H \kappa_{0} \phi_{1} \phi_{2}
 + \frac{1}{2} H (\phi_{1})^2
 + \frac{9}{8} H (\kappa_{0})^2 (\phi_{2})^2,\\
\phi_{1}={}&\frac{1}{2} H \phi_{1} \phi_{2}
 + \tfrac{3}{4} H \kappa_{0} (\phi_{2})^2,\\
\phi_{2}={}&\frac{1}{3} H (\phi_{2})^2.
\end{align*}
The final equation implies that 
\begin{equation*}
    H=\frac{3}{\phi_2}
\end{equation*}
which from \eqref{eq:SpecifyKerrChi} is precisely the condition that 
\begin{equation*}
    H=\frac{6}{\chi}.
\end{equation*}
\end{proof}
\ni Combining Lemma \ref{lem:Kerrstewart} with Proposition \ref{prop:CCKSID} completes the proof of Theorem \ref{thm:CharacterisationOfKerr}.

\section{Precise formulation of the main theorem and concluding remarks}\label{sec:mainthm}

\begin{figure}[t]
\begin{center}
\def\svgwidth{18pc}
\begingroup%
  \makeatletter%
  \providecommand\color[2][]{%
    \errmessage{(Inkscape) Color is used for the text in Inkscape, but the package 'color.sty' is not loaded}%
    \renewcommand\color[2][]{}%
  }%
  \providecommand\transparent[1]{%
    \errmessage{(Inkscape) Transparency is used (non-zero) for the text in Inkscape, but the package 'transparent.sty' is not loaded}%
    \renewcommand\transparent[1]{}%
  }%
  \providecommand\rotatebox[2]{#2}%
  \newcommand*\fsize{\dimexpr\f@size pt\relax}%
  \newcommand*\lineheight[1]{\fontsize{\fsize}{#1\fsize}\selectfont}%
  \ifx\svgwidth\undefined%
    \setlength{\unitlength}{81.46271334bp}%
    \ifx\svgscale\undefined%
      \relax%
    \else%
      \setlength{\unitlength}{\unitlength * \real{\svgscale}}%
    \fi%
  \else%
    \setlength{\unitlength}{\svgwidth}%
  \fi%
  \global\let\svgwidth\undefined%
  \global\let\svgscale\undefined%
  \makeatother%
  \begin{picture}(1,0.66615617)%
    \lineheight{1}%
    \setlength\tabcolsep{0pt}%
    \put(0,0){\includegraphics[width=\unitlength,page=1]{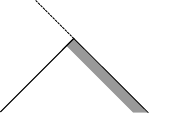}}%
    \put(0.10253134,0.26110836){\color[rgb]{0,0,0}\makebox(0,0)[lt]{\lineheight{1.25}\smash{\begin{tabular}[t]{l}$\mathcal{N}_\star$\end{tabular}}}}%
    \put(0.64508988,0.26023692){\color[rgb]{0,0,0}\makebox(0,0)[lt]{\lineheight{1.25}\smash{\begin{tabular}[t]{l}$\mathscr{I}^+$\end{tabular}}}}%
    \put(0.4290571,0.47272926){\color[rgb]{0,0,0}\makebox(0,0)[lt]{\lineheight{1.25}\smash{\begin{tabular}[t]{l}$\mcS_\star$\end{tabular}}}}%
    \put(0.63363672,0.06857336){\color[rgb]{0,0,0}\makebox(0,0)[lt]{\lineheight{1.25}\smash{\begin{tabular}[t]{l}$\mathcal{V}$\end{tabular}}}}%
  \end{picture}%
\endgroup%

\end{center}
\caption{The existence domain, $\mathcal{V}$ of the past oriented characteristic asymptotic initial value problem.}
\label{fig:narrowrectangle}
\end{figure}
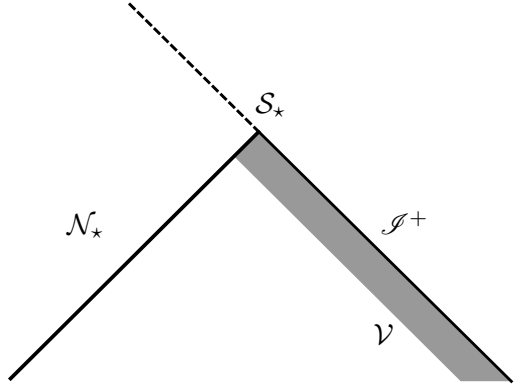

In this section, we use the formalism and notation of the Sections \ref{sec:CIVP} and \ref{Section:KillingSpinorsGeneral} to formulate the precise version of Theorem \ref{thm:mainthmrough}. More precisely, one has that:

\begin{theorem}\label{thm:CharacterisationOfKerr}
Let $\NN_\star$ denote an outgoing null hypersurface in the conformal extension $(\mathcal{M},\bmg,\Xi)$ of a vacuum spacetime $(\tilde{\mathcal{M}},\tilde{\bmg})$ with $\NN_\star\cap\scri\cong\mathcal{S}_\star$. On $\NN_\star\cup\scri$, let initial data for the conformal Einstein field equations in the gauge of Lemmas \ref{lem:SpinCoeffCondition} and \ref{Lemma:ConformalGauge} be posed, as in Figure \ref{fig:narrowrectangle}. The development $\mathcal{V}$ of $\NN_\star\cup\scri$ will possess a Killing spinor if:

 \begin{itemize}
        \item[i. ] On $\mathcal{S}_\star$, $\phi_3=0$ and $\kappa_0$ satisfies the equations
        \begin{align*}
        \begin{aligned}
        \edt \kappa_{0} -\frac{2}{3} \sigma =& 0,\\
        \frac{1}{3}\phi_{0}
        +\kappa_{0} \phi_{1} =& 0, \\
        \frac{2}{3} \phi_{1}
         +3 \kappa_{0} \phi_{2} =& 0,\\
        \end{aligned}
        \end{align*}
        $\kappa_0$ propagated along $\scri$ according to the transport equation
        \begin{equation}
        \Delta \kappa_{0} - \frac{2}{3} \tau- 2 \kappa_{0} \gamma\simeq 0.
        \end{equation}
        Moreover, the dyad components of $B_{ABCD}$ satisfy
        \begin{align}
        \begin{aligned}
         & D B_1 =0, \quad DB_2= 0, \quad DB_3 = 0, \quad DB_4= 0,\\
         & \Delta B_0= 0, \quad \Delta B_1= 0, \quad \Delta B_2 = 0, \quad \Delta B_3= 0.
         \end{aligned}
        \end{align}

        \item[ii. ] On $\scri$ one has that 
        \begin{equation}
        \lambda \simeq 0, \qquad \mu\simeq 0, \qquad \Phi_{02'}\simeq 0, \qquad \phi_4\simeq0;
        \end{equation}
        and the components of $\kappa_{AB}$ satisfy
        \[
        \kappa_1 \simeq -\frac{1}{3}, \qquad \kappa_2\simeq 0.
        \]

         \item[iii.] On $\NN_\star$ one has that
        \begin{align*}
        \begin{aligned}
        & 2 \kappa_{1} (\Phi_{20} - \rho \lambda -  \mu \bar{\sigma}) + \bar{\sigma} \edt \kappa_{2} - 2 \kappa_{2} (\bar\Psi_{1} + \edt \bar{\sigma}) + 3 \lambda \bar{\eth} \kappa_{0}
        + 2 \kappa_{0} \bar{\eth} \lambda - 2 \bar{\eth} \bar{\eth} \kappa_{1}\bumpeq 0,\\
        &B_{ABCD}\bumpeq0.
        \end{aligned}
        \end{align*}
    \end{itemize}
Additionally, $(\tilde{\mathcal{M}},\tilde{\bmg})$ is locally isometric to a portion of a member of the Kerr family of spacetimes with non-vanishing mass if and only if
\begin{equation}\label{eq:KerrCondthm}
\mathrm{Im}\left(\edt\edt\bar{\sigma}\right)=\mathrm{Im}\left(\Phi_{02}\bar{\sigma}\right), \qquad \mbox{on} \quad \mathcal{S}_\star.
\end{equation}
\end{theorem}

\begin{remark}
{\em We make the following remarks on Theorem \ref{thm:CharacterisationOfKerr}:
\begin{itemize}
  \item[a.] Importantly, if any one of the conditions in Theorem \ref{thm:CharacterisationOfKerr} is not satisfied then the portion of $\mathcal{V}$ in $(\tilde{\mathcal{M}},\tilde{\bmg})$ is \emph{not} isomorphic to a portion of the Kerr spacetime.
  \item[b.] Proposition 7 of \cite{ColValRac18}, shows that the Killing vector $\tilde{\xi}_{AA'}$ on $(\tilde{\mathcal{M}},\tilde{\bmg})$ constructed from Theorem \ref{thm:CharacterisationOfKerr} is always Hermitian.
\end{itemize}}
\end{remark}

\subsection*{Concluding remarks}
In this article, we have obtained a characteristaion of the Kerr spacetime which is based on the asymptotic characteristic initial value problem for the conformal Einstein field equations. The characterisation involves conditions 
for a spinor $\kappa_{AB}$ on $\mathcal{N}_\star \cup \scri$. Despite this, the characterisation is constructive. Making use of the algebraic Killing spinor candidate ---see Subsection \ref{Section:AlgebraicKillingSpinorCandidate}--- all the conditions in Theorem \ref{thm:CharacterisationOfKerr} can be verified. Observe that from this perspective, all the conditions on the components of the spinor $\kappa_{AB}$ are, in fact, conditions on the components of the Weyl curvature. 

\medskip
A possible alternative formulation of the characterisation, not considered in this article, is to work in the setting of an asymptotic characteristic problem on a cone. The vertex of this cone would correspond to timelike infinity. While the Kerr spacetime does not have this type of asymptotic structure, there exist, nevertheless, a large class of spacetimes which are Kerr in a neighbourhood of spatial. These spacetimes can be constructed using the Corvino-Schoen gluing method for initial data sets \cite{CorSch03} ---see also \cite{ChrDel02}. This alternative type of asymptotic characterisation of spacetimes with a Kerrian portion of null infinity will be discussed alsewhere. 

\section*{Acknowledgements}
The extensive computations in this article have been carried out using the suite  {\tt xAct} for tensorial manipulations in the Wolfram programming language \cite{xAct}. In particular, we have substantial use of the package {\tt SpinFrames}.

\appendix
\section{The NP  spin coefficient formalism}\label{app:NPgauge}

Consider the spinorial expression 
\begin{equation}
    \nabla_{AA'}\ze^B \label{eq:derivze}
\end{equation}
and the dyad basis
\begin{equation*}
    \ep_{\bf A}{}^A = \{o^A,\iota^A\}
\end{equation*}
normalised to $o_A\io^A=1$, where ${}_{\bf A}$ represents the dyad indices $0,1$. We can write the dyad components of the spinor $\ze$ as $\ze^{\bf A} = \ze^A \ep_{A}{}^{\bmA}$. We want to project \eqref{eq:derivze} in this dyad basis. To this end, let 
\begin{equation*}
     \ep_{\bf A}{}^A \ep_{\bf A'}{}^{A'} \ep_{B}{}^\bmB \nabla_{AA'}\ze^B = \nabla_{\bmA\bmA'}\ze^\bmB+\ze^\bmC\Gamma_{\bmA\bmA'\bmC}{}^\bmB
\end{equation*}
where we define the \emph{Ricci rotation coefficients}, $\Gamma_{\bmA\bmA'\bmC}{}^\bmB\equiv \ep_{B}{}^\bmB\nabla_{\bmA\bmA'}\ep_\bmC{}^\bmB$. Lowering the $\bmB$ index, we find that the Ricci rotation coefficents are symmetric in the final two indices
\begin{equation*}
    \Gamma_{\bmA\bmA'\bmB\bmC} = \Gamma_{\bmA\bmA'\bmC\bmB}.
\end{equation*}

As it is usual, we assign lower case Greek letters to each component of the Ricci rotation coefficients following the NP convention:

\begin{align*}
    \begin{aligned}
        \al \equiv \Gamma_{10'10},&\qquad \be\equiv \Gamma_{01'10,}\\
        \ga \equiv  \Gamma_{11'10},&\qquad\ep\equiv \Gamma_{00'10}, \\
        \pi \equiv  \Gamma_{00'11},&\qquad \tau\equiv \Gamma_{11'00},\\
        \ka \equiv \Gamma_{00'00},&\qquad \nu \equiv \Gamma_{11'11},\\
        \la \equiv  \Gamma_{10'11},&\qquad \si \equiv  \Gamma_{01'00},\\
        \rho \equiv \Gamma_{10'00},&\qquad\mu \equiv \Gamma_{01'11} .
    \end{aligned}
\end{align*}
Moreover, we denote the components of the derivative of $\Xi$ by
\begin{equation}
\Sigma_1\equiv D\Xi, \quad \Sigma_2\equiv \Delta\Xi, \quad \Sigma_3\equiv\delta\Xi,\quad \Sigma_4\equiv\bar{\delta}\Xi.
\end{equation}
Finally, the components of the Weyl spinor $\Psi_{ABCD}$ are 
\[
\{\Psi_0,\Psi_1,\Psi_2,\Psi_3,\Psi_4\}\] and those of the trace-free Ricci spinor $\Phi_{AA'BB'}$ are 
\[\{\Phi_{00},\Phi_{01},\Phi_{02},\Phi_{11},\Phi_{12},\Phi_{22},\Lambda\}.
\] 

\subsection{Ricci equations}

Under the conditions from Lemma \ref{lem:SpinCoeffCondition} the Ricci equations take the form
\begin{subequations}
\begin{align}
D \gamma =& \Lambda- \Phi_{11'}- \Psi_{2}+ \beta \pi+  \alpha \bar{\pi}+  \alpha \tau + \pi \tau+  \beta \bar{\tau} ,\\
 D \tau=& -\Phi_{01'}- \Psi_{1}+ \bar{\pi} \rho+  \pi \sigma+  \rho \tau+ \sigma \bar{\tau},\\
\Delta \pi =& \Phi_{21'}
 + \Psi_{3}
 -  \gamma \pi
 + \bar{\gamma} \pi
 -  \mu \pi
 -  \lambda \bar{\pi}
 -  \lambda \tau
 -  \mu \bar{\tau},\\
\delta \gamma -\Delta \beta=&
 - \Phi_{12'}
 -  \bar{\alpha} \gamma
 - 2 \beta \gamma
 + \beta \bar{\gamma}
 + \alpha \bar{\lambda}
 + \beta \mu
 + \gamma \tau
 + \mu \tau ,\\
 D \beta=&- \Psi_{1}
 +  \beta \rho
 +  \alpha \sigma
 +  \pi \sigma ,\\
D \sigma  =& -\Psi_{0}
 +2 \rho \sigma,\\
 \Delta \mu=&
 \Phi_{22'}
 - \lambda \bar{\lambda}
 - \gamma \mu
 - \bar{\gamma} \mu
 - \mu^2 ,\\
\delta \pi -D \mu= &2 \Lambda
 + \Psi_{2}
 + \bar{\alpha} \pi
 -  \beta \pi
 -  \pi \bar{\pi}
 -  \mu \rho
 -  \lambda \sigma ,\\
\delta \tau-\Delta \sigma =&
 - \Phi_{02'}
 + \bar{\lambda} \rho
 - 3 \gamma \sigma
 + \bar{\gamma} \sigma
 + \mu \sigma
 -  \bar{\alpha} \tau
 + \beta \tau
 + \tau^2 ,\\
\bar\delta \beta - \delta \alpha= &\Lambda
 + \Phi_{11'}
 -  \Psi_{2}
 -  \alpha \bar{\alpha}
 + 2 \alpha \beta
 -  \beta \bar{\beta}
 -  \mu \rho
 + \lambda \sigma ,\label{eq:Riccij}\\
 \bar\delta \gamma- \Delta \alpha  =&
 - \Psi_{3}
 -  \bar{\beta} \gamma
 -  \alpha \bar{\gamma}
 + \beta \lambda
 + \alpha \mu
 + \lambda \tau
 + \gamma \bar{\tau},\\
D \alpha =& -\Phi_{10'}
 +  \alpha \rho
 +  \pi \rho
 +  \beta \bar{\sigma} ,\\
 D \rho  =&- \Phi_{00'}
 + \rho^2
 + \sigma \bar{\sigma},\\
\bar\delta \mu- \delta \lambda =& \Phi_{21'}
 -  \Psi_{3}
 -  \bar{\alpha} \lambda
 + 3 \beta \lambda
 -  \alpha \mu
 -  \bar{\beta} \mu ,\\
 \Delta \lambda  =&
 \Psi_{4}
 - 3 \gamma \lambda
 + \bar{\gamma} \lambda
 - 2 \lambda \mu,\\
\bar\delta \pi-  D \lambda =& \Phi_{20'}
 + 4 \epsilon \lambda
 -  \alpha \pi
 + \bar{\beta} \pi
 -  \pi^2
 -  \lambda \rho
 -  \mu \bar{\sigma}
 ,\\
\bar\delta \sigma-\delta \rho =& \Phi_{01'}
 -  \Psi_{1}
 -  \bar{\alpha} \rho
 -  \beta \rho
 + 3 \alpha \sigma
 -  \bar{\beta} \sigma \label{eq:Riccieqq},\\
\bar\delta \tau- \Delta \rho =& -2 \Lambda
 -  \Psi_{2}
 -  \gamma \rho
 -  \bar{\gamma} \rho
 + \mu \rho
 + \lambda \sigma
 + \alpha \tau
 -  \bar{\beta} \tau
 + \tau \bar{\tau}.
\end{align}\label{eq:ricciidentities}
\end{subequations}

\subsection{The conformal field equations}

In the formalism of Newman-Penrose and the gauge conditions of Lemma \ref{lem:SpinCoeffCondition} the conformal field equations outlined in Section \ref{sec:CFE} have the following form.\\

\subsubsection*{The first conformal field equation (the equation for the conformal factor)}

\begin{subequations}
\begin{align}
\Delta \Sigma_1- \Sigma_1( \gamma +\bar{\gamma} )+ \Sigma_4 \tau + \Sigma_3 \bar{\tau} + \Delta \Sigma_1={}&s
 -  \Xi \Lambda+ \Xi \Phi_{11},\\
 \Delta \Sigma_2+ \Sigma_2 (\gamma +\bar{\gamma})={}&\Xi \Phi_{22},\\
+ \Delta \Sigma_3- \Sigma_3( \gamma - \bar{\gamma}) + \Sigma_2 \tau ={}&\Xi \Phi_{12},\\
D \Sigma_1={}&\Xi \Phi_{00},\\
 D \Sigma_2- \Sigma_3 \pi -  \Sigma_4 \bar{\pi} ={}&s
 -  \Xi \Lambda
 + \Xi \Phi_{11},\\
D \Sigma_3 -  \Sigma_1 \bar{\pi}={}&\Xi \Phi_{01},\\
   \delta \Sigma_1- \Sigma_1 (\bar{\alpha} +  \beta )+ \Sigma_3 \rho + \Sigma_4 \sigma={}&\Xi \Phi_{01},\\
 \delta \Sigma_2+ \Sigma_2( \bar{\alpha} + \beta )-  \Sigma_4 \bar{\lambda} -  \Sigma_3 \mu ={}&\Xi \Phi_{12},\\
 \delta \Sigma_3+\Sigma_3 (\bar{\alpha} -  \beta) -  \Sigma_1 \bar{\lambda} + \Sigma_2 \sigma ={}&\Xi \Phi_{02},\\
\delta \Sigma_4- \Sigma_4 (\bar{\alpha} - \beta) -  \Sigma_1 \mu + \Sigma_2 \rho ={}&- s
 + \Xi \Lambda
 + \Xi \Phi_{11}.
\end{align}
\end{subequations}

\subsubsection*{The second conformal field equation (the equation for the Friedrich scalar)}

\begin{subequations}
\begin{align}
\Delta s={}&\Phi_{22} \Sigma_1
 -  \Phi_{21} \Sigma_3
 -  \Phi_{12} \Sigma_4
 -  \Lambda \Sigma_2
 + \Phi_{11} \Sigma_2,\\
D s={}&- \Lambda \Sigma_1
 + \Phi_{11} \Sigma_1
 -  \Phi_{10} \Sigma_3
 -  \Phi_{01} \Sigma_4
 + \Phi_{00} \Sigma_2,\\
\delta s={}&\Phi_{12} \Sigma_1
 -  \Lambda \Sigma_3
 -  \Phi_{11} \Sigma_3
 -  \Phi_{02} \Sigma_4
 + \Phi_{01} \Sigma_2.
\end{align}
\end{subequations}

\subsubsection*{The third conformal field equation (the Cotton equation)}

\begin{subequations}
\begin{align}
\Delta \Phi_{00} + 2 D \Lambda
 - \delta \Phi_{10}={}&- \bar{\phi}_{2} \Sigma_1
 + \bar{\phi}_{1} \Sigma_3
 - 2 \Phi_{01} \alpha
 - 4 \Phi_{10} \bar{\alpha}
 - 2 \Phi_{10} \beta
 - 2 \Phi_{01} \bar{\beta}
 + 2 \Phi_{00} \gamma\nonumber\\
& + 2 \Phi_{00} \bar{\gamma}
 -  \Phi_{00} \bar{\mu}
 + 2 \Phi_{11} \bar{\rho}
 + \Phi_{20} \sigma,\\
 \Delta \Phi_{01}+ \delta \Lambda
 - \delta \Phi_{11}={}&- \bar{\phi}_{3} \Sigma_1
 + \bar{\phi}_{2} \Sigma_3
 -  \Phi_{02} \alpha
 - 2 \Phi_{11} \bar{\alpha}
 - 2 \Phi_{11} \beta
 -  \Phi_{02} \bar{\beta}
 + 2 \Phi_{01} \gamma\nonumber\\
& -  \Phi_{10} \bar{\lambda}
 -  \Phi_{01} \mu
 + \Phi_{12} \rho
 + \Phi_{21} \sigma,\\
 \Delta \Phi_{02}- \delta \Phi_{12}={}&- \bar{\phi}_{4'} \Sigma_1
 + \bar{\phi}_{3} \Sigma_3
 - 2 \Phi_{12} \beta
 + 2 \Phi_{02} \gamma
 - 2 \Phi_{02} \bar{\gamma}
 - 2 \Phi_{11} \bar{\lambda}
 -  \Phi_{02} \mu\nonumber\\
& + \Phi_{22} \sigma,\\
\Delta \Phi_{11}+ \Delta \Lambda
 - \delta \Phi_{21}={}&- \bar{\phi}_{3} \Sigma_4
 + \bar{\phi}_{2} \Sigma_2
 -  \Phi_{12} \alpha
 -  \Phi_{21} \bar{\alpha}
 + \Phi_{21} \beta
 -  \Phi_{12} \bar{\beta}
 -  \Phi_{20} \bar{\lambda}\nonumber\\
& - 2 \Phi_{11} \mu
 + \Phi_{22} \rho\\
 \Delta \Phi_{12}- \delta \Phi_{22}={}&- \bar{\phi}_{4'} \Sigma_4
 + \bar{\phi}_{3} \Sigma_2
 + \Phi_{22} \bar{\alpha}
 + \Phi_{22} \beta
 - 2 \Phi_{12} \bar{\gamma}
 - 2 \Phi_{21} \bar{\lambda}
 - 2 \Phi_{12} \mu,\\
 D \Phi_{10} - \bar\delta \Phi_{00}={}&\bar{\phi}_{1} \Sigma_1
 -  \bar{\phi}_{0} \Sigma_3
 - 2 \Phi_{00} \alpha
 - 2 \Phi_{00} \bar{\beta}
 + \Phi_{00} \pi
 + 2 \Phi_{10} \rho\nonumber\\
& + 2 \Phi_{01} \bar{\sigma},\\
D \Phi_{02}-\delta \Phi_{01}={}&\phi_{1} \Sigma_3
 -  \phi_{0} \Sigma_2
 - 2 \Phi_{01} \beta
 -  \Phi_{00} \bar{\lambda}
 + 2 \Phi_{01} \bar{\pi}
 + \Phi_{02} \rho\nonumber\\
& + 2 \Phi_{11} \sigma,\\
D \Phi_{11}+ D \Lambda
 - \bar\delta \Phi_{01}={}&\bar{\phi}_{2} \Sigma_1
 -  \bar{\phi}_{1} \Sigma_3
 - 2 \Phi_{01} \alpha
 -  \Phi_{00} \mu
 + \Phi_{01} \pi
 + \Phi_{10} \bar{\pi}
 + 2 \Phi_{11} \rho
 + \Phi_{02} \bar{\sigma},\\
 D \Phi_{12} +\delta \Lambda
 - \delta \Phi_{11}={}&\phi_{2} \Sigma_3
 -  \phi_{1} \Sigma_2
 -  \Phi_{10} \bar{\lambda}
 -  \Phi_{01} \mu
 + \Phi_{02} \pi
 + 2 \Phi_{11} \bar{\pi}
 + \Phi_{12} \rho\nonumber\\
& + \Phi_{21} \sigma,\\
D \Phi_{22}+ 2 \Delta \Lambda
 - \delta \Phi_{21}={}&\phi_{3} \Sigma_3
 -  \phi_{2} \Sigma_2
 + 2 \Phi_{21} \beta
 -  \Phi_{20} \bar{\lambda}
 - 2 \Phi_{11} \mu
 + 2 \Phi_{12} \pi
 + 2 \Phi_{21} \bar{\pi}\nonumber\\
& + \Phi_{22} \rho,\\
\delta \Phi_{10}- \bar\delta \Phi_{01}={}&- \phi_{2} \Sigma_1
 + \bar{\phi}_{2} \Sigma_1
 -  \bar{\phi}_{1} \Sigma_3
 + \phi_{1} \Sigma_4
 - 2 \Phi_{01} \alpha
 + 2 \Phi_{10} \bar{\alpha}
 -  \Phi_{20} \sigma\nonumber\\
& + \Phi_{02} \bar{\sigma},\label{eq:CFE3k}\\
\delta \Phi_{11}+  \delta \Lambda
 - \bar\delta \Phi_{02}={}&\bar{\phi}_{3} \Sigma_1
 -  \phi_{2} \Sigma_3
 -  \bar{\phi}_{2} \Sigma_3
 + \phi_{1} \Sigma_2
 - 2 \Phi_{02} \alpha
 + 2 \Phi_{02} \bar{\beta}
 + \Phi_{10} \bar{\lambda}\nonumber\\
& -  \Phi_{01} \mu
 + \Phi_{12} \rho
 -  \Phi_{21} \sigma,\label{eq:CFE3l}\\
\delta \Phi_{21}- \bar\delta \Phi_{12}={}&- \phi_{3} \Sigma_3
 + \bar{\phi}_{3} \Sigma_4
 + \phi_{2} \Sigma_2
 -  \bar{\phi}_{2} \Sigma_2
 - 2 \Phi_{21} \beta
 + 2 \Phi_{12} \bar{\beta}
 -  \Phi_{02} \lambda\nonumber\\
& + \Phi_{20} \bar{\lambda}.
\end{align}
\end{subequations}

\subsubsection*{The fourth conformal field equation (the Bianchi identity)}

\begin{subequations}
\begin{align}
\Delta \phi_{0} - \delta \phi_{1}={}&-2 \phi_{1} \beta
 + 4 \phi_{0} \gamma
 -  \phi_{0} \mu
 + 3 \phi_{2} \sigma
 - 4 \phi_{1} \tau, \label{eq:CFE4a}\\ 
\Delta \phi_{1}- \delta \phi_{2}={}&2 \phi_{1} \gamma
 - 2 \phi_{1} \mu
 + 2 \phi_{3} \sigma
 - 3 \phi_{2} \tau,\\
\Delta \phi_{2}-\delta \phi_{3}={}&2 \phi_{3} \beta
 - 3 \phi_{2} \mu
 + \phi_{4} \sigma
 - 2 \phi_{3} \tau,\\
\Delta \phi_{3}- \delta \phi_{4}={}&4 \phi_{4} \beta
 - 2 \phi_{3} \gamma
 - 4 \phi_{3} \mu
 -  \phi_{4} \tau,\\
D \phi_{1}- \bar\delta \phi_{0}={}&-4 \phi_{0} \alpha
 + \phi_{0} \pi
 + 4 \phi_{1} \rho,\\
D \phi_{2}- \bar\delta \phi_{1}={}&-2 \phi_{1} \alpha
 -  \phi_{0} \lambda
 + 2 \phi_{1} \pi
 + 3 \phi_{2} \rho,\\
D \phi_{3}- \bar\delta \phi_{2}={}&
 - 2 \phi_{1} \lambda
 + 3 \phi_{2} \pi
 + 2 \phi_{3} \rho,\\
D \phi_{4}- \bar\delta \phi_{3}={}&2 \phi_{3} \alpha
 - 3 \phi_{2} \lambda
 + 4 \phi_{3} \pi
 + \phi_{4} \rho.\label{eq:CFE4b}
\end{align}
\end{subequations}

\bibliographystyle{reporthack}

\end{document}